\documentclass[sigconf]{acmart}

\renewcommand\footnotetextcopyrightpermission[1]{} 
\usepackage{amsmath,amsfonts}
\usepackage{algorithmic}
\usepackage{flushend}
\usepackage{comment}
\usepackage[ruled,vlined]{algorithm2e}
\def\checkmark{\tikz\fill[scale=0.4](0,.35) -- (.25,0) -- (1,.7) -- (.25,.15) -- cycle;} 
\usepackage{pifont}
\newcommand{\xmark}{\ding{55}}

\usepackage{tikz}
\usepackage{graphicx}
\usepackage{amsthm}
\usepackage{amsmath,amsfonts}
\usepackage{graphicx}
\usepackage{textcomp}
\usepackage{tikz}
\usepackage{tikz-qtree}
\usepackage{epstopdf}

\usetikzlibrary{mindmap,trees}
\usetikzlibrary{arrows}
\usetikzlibrary{calc}
\usepackage{etoolbox}

\usepackage{graphicx}
\usepackage[font={bf,it,small},labelsep=colon]{caption}
\usepackage{subcaption}
\usepackage{xspace}
\usepackage{color}
\usepackage{enumitem}
\usepackage{flushend}
\usepackage{comment}
\usetikzlibrary{tikzmark}
\usepackage[hang,flushmargin]{footmisc} 
\usepackage[none]{hyphenat}

\newcommand{\name}{{\sc Morshed}\xspace}

\AtBeginDocument{%
  \providecommand\BibTeX{{%
    \normalfont B\kern-0.5em{\scshape i\kern-0.25em b}\kern-0.8em\TeX}}}

\iffalse
\newcommand{\mustafa}[1]{\todo[inline,color=blue!40]{mustafa: #1}}
\newcommand{\saurabh}[1]{\todo[inline,color=red!40]{saurabh: #1}}
\newcommand{\issa}[1]{\todo[inline,color=orange!40]{issa: #1}}
\else 
\newcommand{\mustafa}[1]{}
\newcommand{\saurabh}[1]{}
\newcommand{\issa}[1]{}
\fi

\newcommand{\pn}[2]{{\color{red}{Parinaz (#1): #2}}}

\begin{document}
\sloppy

\title{\textit{\name}: Guiding Behavioral Decision-Makers towards Better Security Investment in Interdependent Systems}

\author{Mustafa Abdallah}
\email{abdalla0@purdue.edu}
\affiliation{%
  \institution{Purdue University}
}

\author{Daniel Woods}
\email{woods104@purdue.edu}
\affiliation{%
  \institution{Purdue University}
}

\author{Parinaz Naghizadeh}
\email{naghizadeh.1@osu.edu}
\affiliation{%
  \institution{Ohio State University}
}

\author{Issa Khalil}
\email{ikhalil@hbku.edu.qa}
\affiliation{%
  \institution{Qatar Computing Research Institute}
}

\author{Timothy Cason}
\email{cason@purdue.edu}
\affiliation{%
  \institution{Purdue University}
}

\author{Shreyas Sundaram}
\email{sundara2@purdue.edu}
\affiliation{\institution{Purdue University}}

\author{Saurabh Bagchi}
\email{sbagchi@purdue.edu}
\affiliation{\institution{Purdue University}}

\renewcommand{\shortauthors}{Trovato and Tobin, et al.}

\begin{abstract}
We model the {\em behavioral} biases of human decision-making in securing interdependent systems and show that such behavioral decision-making leads to a suboptimal pattern of resource allocation compared to non-behavioral (rational) decision-making. We provide empirical evidence for the existence of such behavioral bias model through a controlled subject study with 145 participants. 
We then propose three learning techniques for enhancing decision-making in multi-round setups. 
We illustrate the benefits of our decision-making model through multiple interdependent real-world systems and quantify the level of gain compared to the case in which the defenders are behavioral. We also show the benefit of our learning techniques against different attack models. We identify the effects of different system parameters on the degree of suboptimality of security outcomes due to behavioral decision-making. \end{abstract}



\maketitle
\section{Introduction}
Most of the current IT-based systems are becoming more complex, however they are facing sophisticated attacks from external adversaries where the attacker's goal is to breach specific (critical) assets within the system. 
For each critical asset, the attacker typically utilizes different vulnerabilities to compromise such asset. 
In this context, the system operators, Chief Information Security Officer or security executives have to judiciously allocate their (often limited) security budgets to reduce security risks of the systems they manage. This resource allocation problem is further complicated by the fact that a large-scale system consists of multiple interdependent subsystems managed by different operators, with each operator in charge of securing her own subsystem. 

Prior work has considered such security decision-making problems in both decision-theoretic and game-theoretic settings~\cite{laszka2015survey,yan2012towards} in which the security risk faced by an operator (defender) depends on her security investments. However, most of the existing work relied on \emph{classical models} of decision-making, where all defenders and attackers are assumed to make fully rational risk evaluations and security decisions~\cite{laszka2015survey, hota2016optimal,modelo2008determining}. 

On the contrary, behavioral economics has shown that humans consistently deviate from these classical models of decision-making. Most notably, research in {\em behavioral economics}, has shown that humans perceive gains, losses and probabilities in a skewed, nonlinear manner~\cite{kahneman1979prospect}. In particular, humans typically overweight low probabilities and underweight high probabilities, where this weighting function has an inverse S-shape, as shown in Figure~\ref{fig:Prelec Probability weighting function}. Many empirical studies (e.g., \cite{gonzalez1999shape,kahneman1979prospect}) have provided evidence for this class of behavioral models. These effects are relevant for evaluating security of such systems in which decisions on implementing security controls are not made purely by automated algorithms, but rather through human decision-making, albeit with help from threat assessment tools \cite{sheyner2002automated,jauhar2015model}. 

There are many articles discussing the prevalence of human factors in security decision-making, both in popular press and in academic journals~\cite{dor2016model}, none of which however shed light on the impact of cognitive biases on the overall system security and how we can mitigate such biases. Our work bridges this gap by showing how behavioral research can lead to better security decision-making for interdependent systems. Specifically, we study the effect of the aforementioned human behavioral decision-making bias on security allocations and propose multiple techniques to overcome such bias in both single-round and multi-round setups. 

There are recent  works~\cite{7544460,sanjab2017prospect} that have started to leverage mathematical analysis to model and predict the effect of behavioral decision-making on the players' investments.  However, these works have the following limitations. First, they have considered the impact of probability weighting in certain specific classes of interdependent security games. Second, these works did not consider multiple-round setups in which defenders can learn. In contrast to those, we consider general defense allocation techniques that can be applied to any system where its failure scenarios are modeled by an attack graph, and we propose multi-round learning algorithms to guide behavioral decision-makers in different setups  and consider different types of attackers.  The difference between \name \footnote{Morshed is an Arabic word with the meaning of guiding people to the right place.} and previous related work is shown in Table~\ref{tbl:morshed_related_work}.


\noindent {\bf Our contributions}: \\
In this paper, we first study the effects of human behavioral decision-making on the security of interdependent systems with multiple defenders where each defender is responsible for defending a set of assets (i.e., a subnetwork of the whole system network). In interdependent systems, stepping-stone attacks are often used by external attackers to exploit vulnerabilities within the network in order to reach and compromise critical targets.  These stepping-stone attacks can be captured via {\it attack graphs}, representing all possible paths an attacker may take to reach targets within the system~\cite{modelo2008determining,homer2013aggregating}. 

We design a reasoning and security investment decision-making technique that we call \name 
pronounced as \textit{M-or-Sh-ed}. 
We first describe the model consisting of multiple behavioral defenders and an attacker, in which the interdependencies between the defenders' assets are captured via an attack graph by proposing a \emph{behavioral security game} model. We show that behavioral decision-making leads to suboptimal resource allocation compared to non-behavioral decision-making. We then propose different learning-based techniques for guiding behavioral decision-makers  towards optimal investment decisions for two different scenarios where each scenario represents whether the defender has knowledge of the adversary's history (i.e., chosen attack paths in previous rounds) or not. Our proposed techniques enhance the implemented security policy (in terms of reducing the total system loss when compromised by allocating limited security resources optimally).
\name has components for both single-round and multi-round setups as shown in Figure~\ref{fig:System_High_level_overview}. We consider two classes of defenders. 

\noindent\textit{\bf Behavioral defenders}: These defenders make security investment decisions under two types of cognitive biases. First, following prospect-theoretic,  non-linear probability weighting models, they misperceive the probabilities of a successful attack on each edge of the attack graph. Second, they have a bias toward spreading their budget so that a minimum, non-zero investment is allocated to each edge of the attack graph. This second kind of bias is motivated by behavior that we observe in our human subject experiments (see Section~\ref{sec:human-experiments}). 

\noindent \textit{\bf Non-behavioral or rational defenders}: These defenders make security investment decisions based on the classical models of fully rational decision-making. Specifically, they correctly perceive the risk on each edge within the attack graph of the system network.

On the other hand, almost all research that have considered behavioral economics in security and privacy has the common theme of considering individual choices regarding privacy and how people treat their own personal data~\cite{acquisti2009nudging} or  entirely based on psychological studies \cite{anderson2012security}. To the best of our knowledge, none of these research considered the defense choices made by people in organizational contexts with interdependent system under control. On the contrary, our work considers scenarios that can be applied to critical infrastructure systems (e.g., cyber-physical systems).
\begin{table}[t]
\caption {Comparison between the prior related work and \name in terms of the available features.}
\label{tbl:morshed_related_work}
\centering
\resizebox{\columnwidth}{!}
{%
\begin{tabular}{|l|l|l|l|l|l|l|}
\hline
\multicolumn{1}{|l|}{\text{\bf System}}
& \multicolumn{1}{l|}{\bf \shortstack{Multiple \\ Defenders}}
& \multicolumn{1}{l|}{\bf \shortstack{Interdependent \\ subnetworks}}
& \multicolumn{1}{l|}{\bf \shortstack{Analytical \\ Framework}}
& \multicolumn{1}{l|}{\bf \shortstack{Behavioral \\ Biases}}
& \multicolumn{1}{l|}{\bf \shortstack{Various Attack \\ Types}}
& \multicolumn{1}{l|}{\bf \shortstack{Multiple \\ Rounds}}\\
\cline{1-6}
\hline
RAID08~\cite{modelo2008determining}, MILCOM06~\cite{lippmann2006validating} & \xmark  & \checkmark & \xmark & \xmark  & \xmark & \xmark\\
\hline
S\&P02~\cite{sheyner2002automated},  CCS12~\cite{yan2012towards}  & \xmark & \xmark & \checkmark & \xmark  & \xmark & \xmark \\
\hline
S\&P09~\cite{acquisti2009nudging}, EC18~\cite{redmiles2018dancing}, ACSAC12~\cite{anderson2012security} & \xmark & \xmark & \xmark & \checkmark & \xmark & \xmark \\  
\hline
ICC17~\cite{sanjab2017prospect} & \xmark  & \checkmark  &  \checkmark & \checkmark & \xmark & \xmark \\
\hline
TCNS20~\cite{abdallah2020behavioral}, TCNS18~\cite{7544460} & \checkmark & \checkmark & \checkmark & \checkmark & \xmark & \xmark \\
\hline
\name & \checkmark & \checkmark & \checkmark & \checkmark & \checkmark & \checkmark \\
\hline
\end{tabular}}
\vspace{-0.1in}
\end{table}

We perform a human subject study with N = 145 participants where they choose defense allocations in two simple attack graphs. We then evaluate \name using five synthesized attack graphs that represent realistic interdependent systems and attack paths through them. These systems are DER.1~\cite{jauhar2015model}, (modelled by NESCOR), SCADA industrial control system, modeled using NIST guidelines for ICS~\cite{hota2016optimal}, IEEE 300-bus smart grid~\cite{khanabadi2012optimal}, E-commerce~\cite{modelo2008determining}, and VOIP~\cite{modelo2008determining}. 
We do a benchmark comparison with two prior solutions for optimal security controls with attack graphs \cite{sheyner2002automated, lippmann2006validating}, and quantify the level of the underestimation of loss compared to the \name evaluation where defenders are behavioral. In conducting our analysis and obtaining these results based on a behavioral model, we address several domain-specific challenges in the context of security of interdependent systems. These include augmenting the attack graph with certain parameters such as sensitivity of edges to security investments (Equation~\ref{eq:expon_prob_func}), the estimation of baseline attack probabilities (Table~\ref{tbl:cvss_cve_der_scada}) and the types of defense mechanisms (Section~\ref{sec: eval_multiple_def}) in our formulations.
\begin{figure}
\centering
  \includegraphics[width=0.9\linewidth]{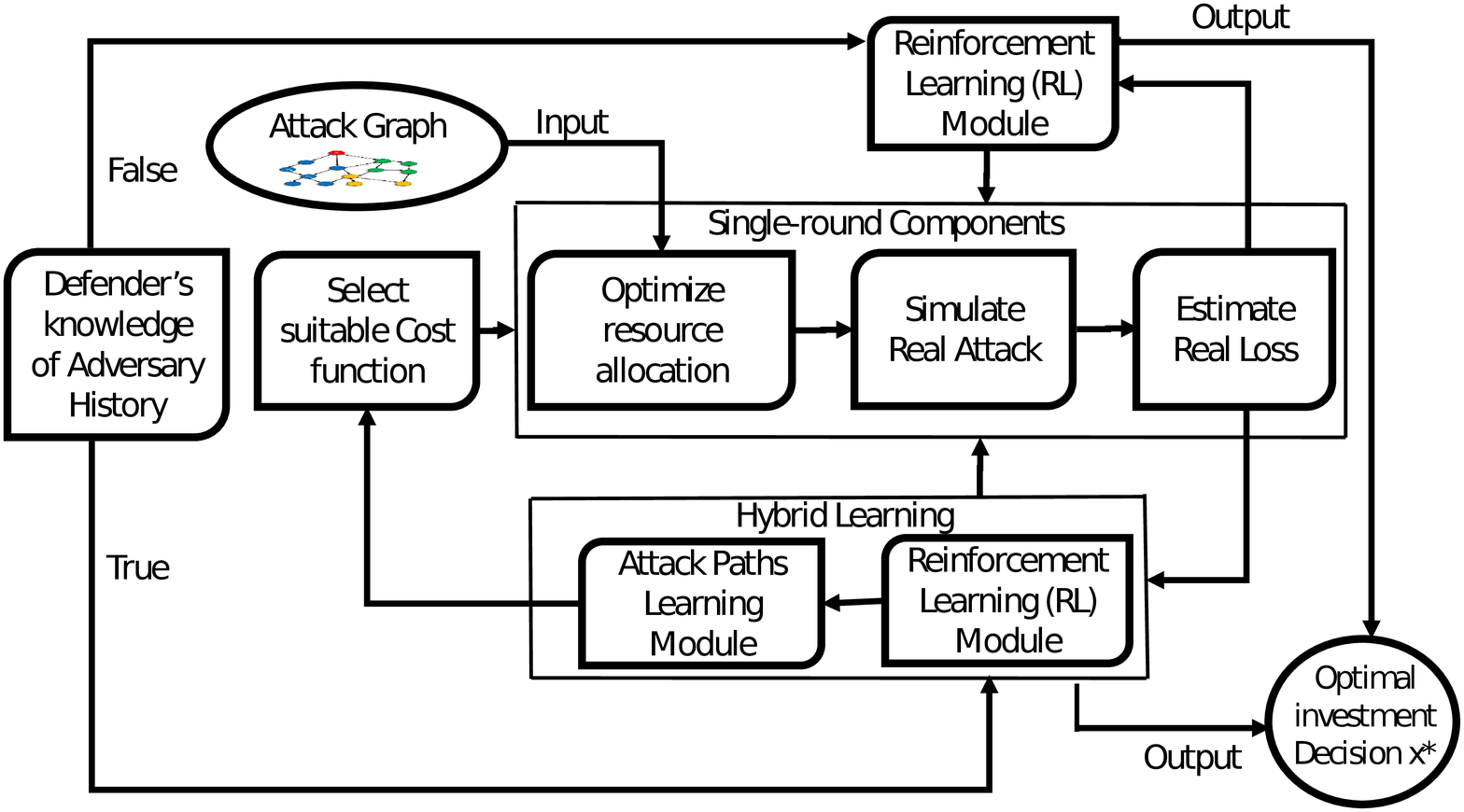}
  \caption{A high level overview of \name's flow,  available features and main components (e.g., single-round and Hybrid Learning).}
   \label{fig:System_High_level_overview}
\vspace{-0.1in}
\end{figure}

In summary, this paper makes the following contributions: 
\vspace{-3pt}

\begin{enumerate}[noitemsep,topsep=2pt,parsep=0pt,partopsep=2pt,leftmargin=*]
\item We propose a \textit{security investment guiding} technique for the defenders of interdependent systems where defenders' assets have mutual interdependencies. We show the effect of {\em behavioral} biases of human decision-making on system security and we quantify the level of gain due to our decision-making technique where defenders are behavioral.

\item We validate the existence of bias via a controlled subject study and illustrate the benefits of our decision-making through multiple real-world interdependent systems. We also analyze the different system parameters that affect the security of interdependent systems under our behavioral model. 

\item We propose three learning techniques to improve defense decisions in multi-round scenarios against different attack models that affect the security of interdependent systems. 
We incorporate such effects with behavioral decision-making.
\end{enumerate}

\section{Background and PROBLEM SETUP}
\label{sec:model}
We begin by presenting a background on behavioral security games, establishing a theoretical basis that can be used to model any multi-defender interdependent system. A simple example of our setup is shown in Figure \ref{fig:Overview attack Graph}, which represents a system consisting of 3 interdependent defenders. An external attacker aims to exploit vulnerabilities within the network in order to reach and compromise critical targets \cite{hota2016optimal,jauhar2015model}. We formalize the attacker and defenders' goals and actions in this section. 

\begin{figure}[t]
\begin{center}
  \includegraphics[width=0.6\linewidth]{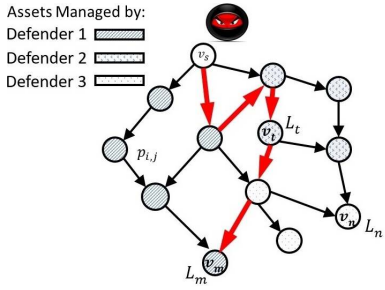}
  \caption{Overview of the interdependent security framework. The interdependencies between assets are represented by edges. 
  An attacker tries to compromise critical assets using stepping stone attacks starting from node $v_s$. The bold (red) edges show one such attack path.}
  \label{fig:Overview attack Graph}
\end{center}
\vspace{-0.08in}
\end{figure}

\subsection{Threat Model}\label{sec:threat-model}

We study security games consisting of one attacker and multiple defenders interacting through an attack graph ${G}=({V}, \mathcal{E})$. The nodes ${V}$ of the attack graph represent the assets in the system, while the edges $\mathcal{E}$ capture the attack progression between the assets. In particular, an edge from $v_i$ to $v_j$, $(v_i,v_j)\in \mathcal{E}$, indicates that if asset $v_i$ is compromised by the attacker, it can be used as a stepping stone to launch an attack on asset $v_j$ (e.g., if an attacker gains the password required to access a power plant's control software ($v_i$), it can use it to attempt to alter the operation of a generator ($v_j$)). The baseline probability that the attacker can successfully compromise $v_j$ given that it has compromised $v_i$, is denoted by the edge weight $p^0_{i,j}\in [0,1]$. By ``baseline probability'' we mean the probability of successful compromise without any security investment in protecting the assets. The attacker initiates attacks on the network from a source node $v_s$ (or multiple possible source nodes), and aims to reach a target node $v_m \in V_k$, i.e., a critical node for defender $D_k$.

\subsection{Defense Model}\label{sec:defense-model}
Each defender $D_{k}\in D$ is in control of a subset of assets $V_k\subseteq V$. This is motivated by the fact that a large system comprises a number of smaller subnetworks, each owned by an independent stakeholder. Among all the assets in the network, a subset $V_m\subseteq V$ are \emph{critical} assets, the compromise of which entails a financial loss for the corresponding defender. Specifically, if asset $v_m\in V_m$ is compromised by the attacker, any defender $D_{k}$ for whom $v_m\in V_k$ suffers a financial loss $L_m\in \mathbb{R}_{>0}$. 

To protect the critical assets from being reached through stepping stone attacks, the defenders can choose to invest their resources in strengthening the security of the edges in the network. Specifically, let $x^k_{i,j}$ denote the investment of a defender $D_k$ on edge $(v_i,v_j)\in \mathcal{E}_k$, and let $x_{i,j}=\sum_{D_k \in D} x^k_{i,j}$ be the total investment on that edge by all eligible defenders. Then, the probability of successfully compromising $v_j$ starting from $v_i$ is given by $p_{i,j}(x_{i,j})$. In addition, let $s_{i,j}\in[1,\infty)$ denote the sensitivity of edge $(v_i,v_j)$ to the total investment $x_{i,j}$. For larger sensitivity values, the probability of successful attack on the edge decreases faster with each additional unit of security investment on that edge; in other words, edges that are easier to defend will have larger sensitivity. 

Let $P_m$ be the set of all attack paths from $v_s$ to $v_m$. The defender assumes the worst-case scenario, i.e., the attacker\footnote{Our formulation also captures the case where each defender faces a different attacker who exploits the most vulnerable path from the source to that defender's assets.} exploits the most vulnerable path to each target.%
\footnote{We will consider different types for the attacker (with partial knowledge) in Section~\ref{sec:learning_rounds}.} Note that previous works considered such adversary model that chooses the most vulnerable path to target assets (e.g.,~\cite{hota2016optimal,laszka2015survey}).
Mathematically, this can be captured via the following total loss function for $D_k$: 
\begin{equation}\label{eq:defender_utility}
\hat{C}_{k}(\mathbf{x}) = \sum_{v_{m} \in V_{k}} L_{m} \hspace{0.3mm} \Big( \hspace{0.3mm} \underset{P \in P_{m}}{\text{max}}\prod_{(v_{i},v_{j}) \in P} p_{i,j}(x_{i,j}) \hspace{0.3mm} \Big)~.
\end{equation}

We let the probability of successfully compromising $v_j$ starting from $v_i$ be given by, 
\begin{equation}\label{eq:expon_prob_func}
p_{i,j}(x_{i,j})= p_{i,j}^0\exp{\Big(-  s_{i,j} \sum_{D_{k} \in {D} \text{ s.t. } (v_i,v_j)\in \mathcal{E}_k} { x^{k}_{i,j}}\Big)}.
\end{equation}
That is, the probability of successful attack on an edge $(v_i,v_j)$ decreases exponentially with the sum of the investments on that edge by all defenders. This probability function falls within a class commonly considered in security economics (e.g., \cite{gordon2002economics,hota2016optimal}). 

%
\begin{figure}[t]
\begin{center}
  \includegraphics[width=0.59\linewidth]{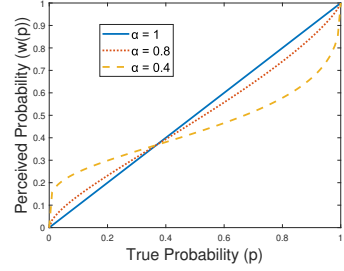}
  \caption{Prelec Probability weighting function which transforms true probabilities $p$  into perceived probabilities $w(p)$. The parameter $\alpha$ controls the extent of overweighting and underweighting, with $\alpha = 1$ indicating non-behavioral or rational decision-making.}
  \label{fig:Prelec Probability weighting function}
\end{center}
\vspace{-0.083in}
\end{figure}

\subsection{Behavioral Probability Weighting}
As mentioned in the Introduction, the  behavioral economics literature has shown that humans consistently misperceive probabilities by overweighting low probabilities, and underweighting high probabilities \cite{kahneman1979prospect,prelec1998probability}.  More specifically, many humans perceive a ``true'' probability $p$ as probability $w(p)$, where $w(\cdot)$ is known as a \emph{probability weighting function}.  A commonly studied functional form for this weighting function was formulated by Prelec in \cite{prelec1998probability}, shown in Figure~\ref{fig:Prelec Probability weighting function}, and is given by
\begin{equation}\label{eq:prelec}
w(p) = \exp{\Big[-(-\log(p)\hspace{0.2mm})^{\alpha}\hspace{0.5mm} \Big]},  \hspace{3mm} p\in [0,1],
\end{equation}
where $\alpha \in (0,1]$ is a parameter that controls the extent of misperception. 
When $\alpha = 1$, we have $w(p) = p$ for all $p \in [0,1]$, which corresponds to the situation where probabilities are perceived correctly, i.e., a non-behavioral defender.  

\subsection{Perceived Costs of a Behavioral Defender} 
We now incorporate this probability weighting function into the security game of Section \ref{sec:defense-model}. 
In a \emph{behavioral security game}, each defender misperceives the attack success probability on {each edge} according to the probability weighting function in \eqref{eq:prelec}. 
She then chooses her investments $x_k:=\{x^k_{i,j}\}_{(v_i,v_j)\in \mathcal{E}_k}$ to minimize her \emph{perceived} loss
\begin{equation}\label{eq:defender_utility_edge}
C_{k}(x_{k},\mathbf{x}_{-k}) = \sum_{v_{m} \in V_{k}} L_{m} \Big(  \underset{P \in P_{m}}{\text{max}}\prod_{(v_{i},v_{j}) \in P} w\left(p_{i,j}(x_{i,j}) \right) \Big)~,
\end{equation}
subject to her total security investment budget $B_k$, i.e., $\sum_{(v_i,v_j)\in \mathcal{E}_k} x^k_{i,j} \leq B_k$, and non-negativity of the investments, i.e., $x^k_{i,j}\geq 0$. We prove the convexity of the total loss~\eqref{eq:defender_utility_edge} in  Appendix~\ref{sec:proof-convexity}.

\begin{figure*}[t] 
\centering
\begin{subfigure}[t]{.48\textwidth}
\centering
\begin{tikzpicture}[scale=0.5]

\tikzset{edge/.style = {->,> = latex'}};

\node[draw,shape=circle] (vs) at (-2,0) {$v_s$};
\node[draw,shape=circle] (v1) at (0,0) {$v_1$};
\node[draw,shape=circle] (v2) at (2,1) {$v_2$};
\node[draw,shape=circle] (v3) at (2,-1) {$v_3$};
\node[draw,shape=circle] (v4) at (4,0) {$v_4$};
\node[draw,shape=circle] (v5) at (6,0) {$v_5$};
\node[shape=circle] (L1) at (6,1) {$L_5=1$};

\draw[edge,thick] (vs) to (v1);
\draw[edge,thick] (v1) to (v2);
\draw[edge,thick] (v1) to (v3);
\draw[edge,thick] (v2) to (v4);
\draw[edge,thick] (v3) to (v4);
\draw[edge,thick] (v4) to (v5);
\end{tikzpicture}
\caption{An attack graph with a min-cut edge.} 
\label{fig:split_join_dependence_before}
\end{subfigure}
\begin{subfigure}[t]{.48\textwidth}
\centering
\begin{tikzpicture}[scale=0.5]

\tikzset{edge/.style = {->,> = latex'}};

\node[draw,shape=circle] (vs) at (0,0) {$v_1$};
\node[draw,shape=circle] (v1) at (2,1) {$v_2$};
\node[draw,shape=circle] (v2) at (2,-1) {$v_3$};
\node[draw,shape=circle] (v3) at (4,0) {$v_4$};

\node[shape=circle] (L1) at (4,1) {$L_4=1$};

\draw[edge,thick] (vs) to (v1);
\draw[edge,thick] (vs) to (v2);
\draw[edge,thick] (v1) to (v2);
\draw[edge,thick] (v1) to (v3);
\draw[edge,thick] (v2) to (v3);
\end{tikzpicture}
\caption{An attack graph with a cross-over edge.} 
\label{fig:cross_over_edge_graph}
 \end{subfigure} 
\caption{The attack graph in (a) is used to illustrate the sub-optimal investment decisions of behavioral defenders. The attack graph in (b) is used in the human subject experiment to isolate the spreading effect.}
\label{fig:MPNE_main}
\vspace{-2mm}
\end{figure*}
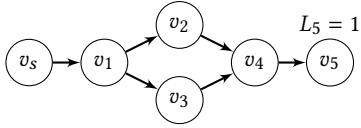
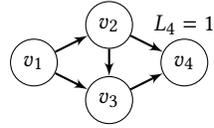

\subsection{Spreading Nature of Security Investments}
We augment our model with another aspect of behavioral decision-making, which we call {\em spreading}. A defender with this characteristic spreads some of her investments on 
all edges of the attack graph, even when some edges are unlikely to be exploited for attacks. 
Spreading here is inspired by {\em Na\"ive Diversification}~\cite{benartzi2001naive} 
from behavioral economics, where humans have a tendency to split investments evenly over the available options.
This phenomenon has not been reported earlier for security decision-making, to the best of our knowledge, and we infer this behavior from our human subject study (detailed in Section~\ref{sec:human-experiments}).
We capture this effect by adding another constraint to our model in (\ref{eq:defender_utility_edge}): for each defender $D_k$, we set $x^k_{i,j} \geq \eta_k$, where $\eta_k$ is the minimum investment $D_k$ makes on any edge. The value $\eta_{k} = 0$ gives us the behavioral decision with no spreading, i.e., with only behavioral probability weighting.  

\section{Human Subject Study}\label{sec:human-experiments}

\begin{figure*}[t] 
\begin{minipage}[t]{1.0\textwidth}
\begin{minipage}[t]{.47\textwidth}
\centering
   \includegraphics[width=\linewidth]{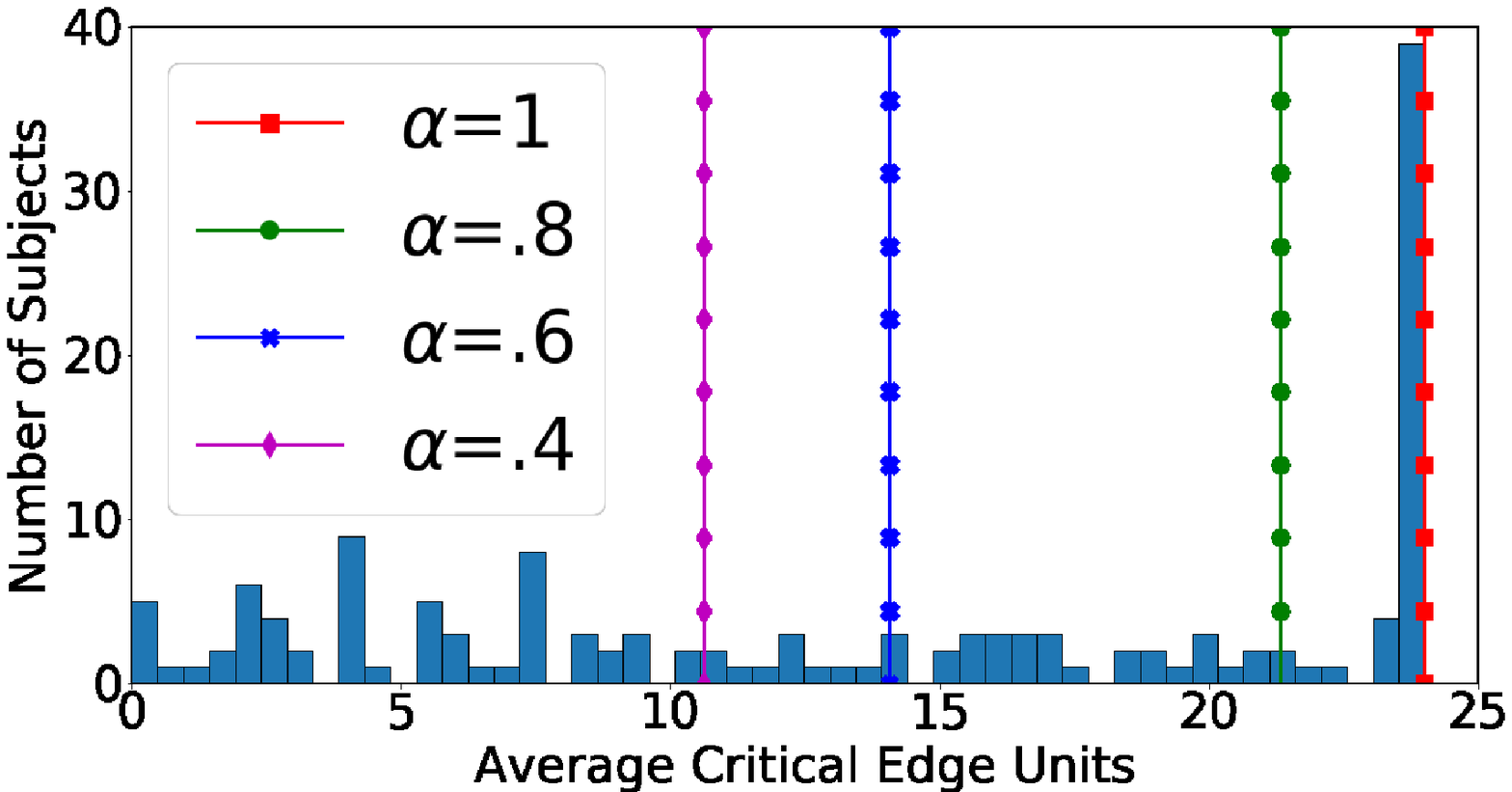}
  \caption{Subjects' investments on the critical edge. Vertical lines with dots show optimal allocations at specific behavioral levels ($\alpha$).}
  \label{fig:Avg_crit_red}
\end{minipage}\hfill
\begin{minipage}[t]{.47\textwidth}
\centering
   \includegraphics[width=\linewidth,height=55mm,keepaspectratio]{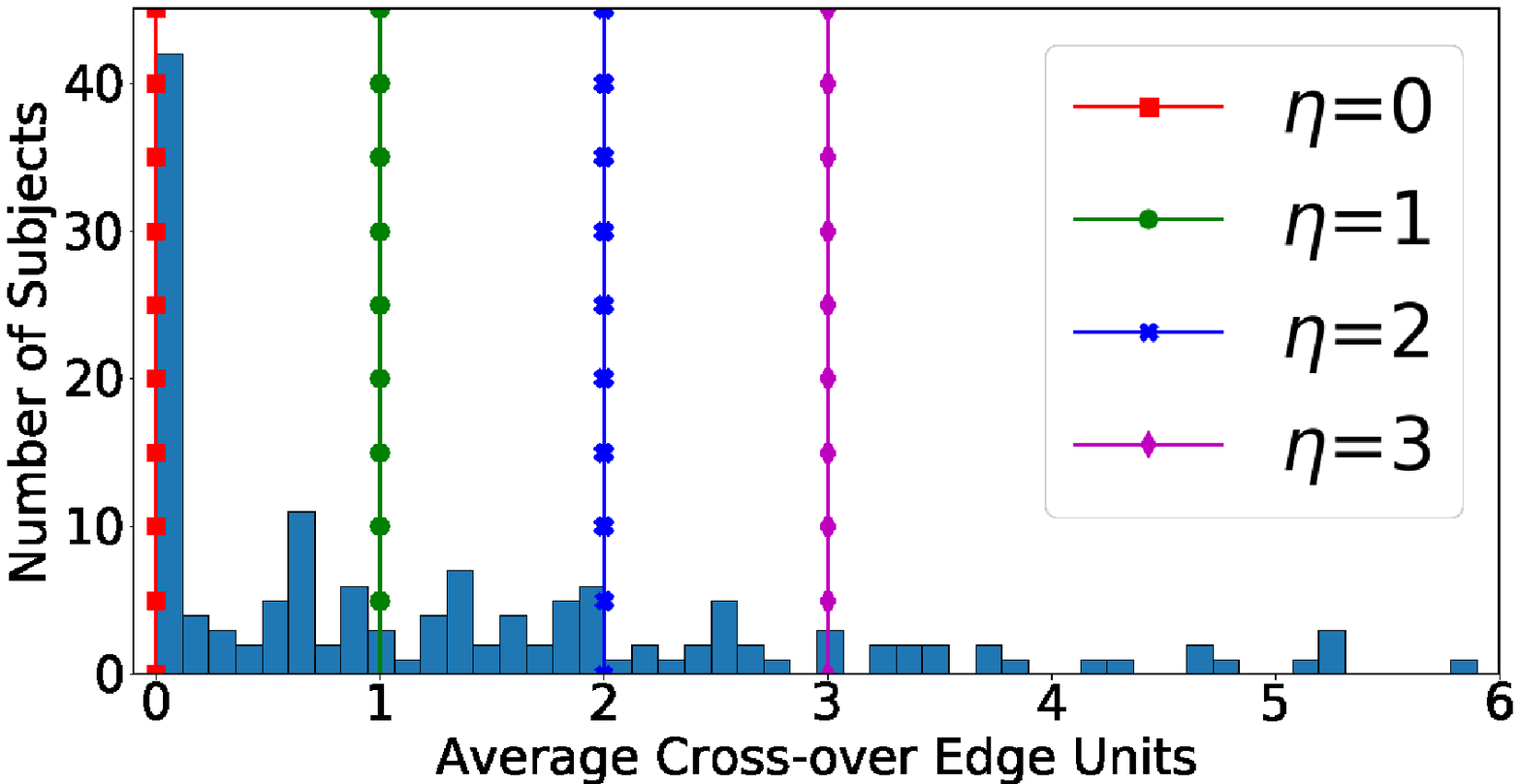}
  \caption{Subjects' investments on the cross-over edge. Vertical lines with dots show optimal allocations at specific spreading levels ($\eta$).}
  \label{fig:Avg_crit_blue}
\end{minipage} 
  \end{minipage}
  \vspace{-1mm}
\end{figure*}

To validate the existence of behavioral bias in security allocations (captured by our model in Section~\ref{sec:model}), incentivized experiments were conducted on 145 students in an Experimental Economics Laboratory at a large public university. Subject demographics are presented in Appendix~\ref{app: human-exp-extended}. 
Subjects participated in the role of a defender, and allocated 24 discrete defense units over edges in each network. 
Subjects made their decisions on a computerized interface, and faced 10 rounds for each network, receiving feedback after each round indicating whether the attack was successful or not (i.e., the asset was compromised).
Subjects received comprehensive written instructions on the decision environment that explained how their investment allocation mapped into the probability of edge defense, and what was considered a successful defense. Subjects received a base payment of \$5.00 for their participation. In addition, we randomly selected one round from each network and if the subject successfully defended the critical node in that round, she received an additional payment of \$7.50. 
 
\subsection{Network (A) with Critical Edge}
\label{sec:human-experiments-A}
This human experiment is on a network similar to Figure \ref{fig:split_join_dependence_before}, except that there is only one critical edge $(v_4,v_5)$ i.e., $v_s = v_1$. Figure \ref{fig:Avg_crit_red} shows the average investment allocation to the critical edge, based on 1450 investment decisions (i.e., 10 decisions 
from each of the 145 subjects). It shows the proportion of subjects who are non-behavioral (those at the vertical red line of $\alpha = 1$, 27\%), as well as heterogeneity in $\alpha$, with observations further to the left being more behavioral. 
Subjects to the left of the $\alpha=0.4$ line (approximately 10 units allocated to the critical edge) are not necessarily exhibiting $\alpha < 0.4$. Those who allocate between 5 and 10 units to the critical edge could have a strong preference for spreading. 
  We observe that after round 4, the average investment on the critical edge in each round is higher than the initial investment in round 1 (Figure~\ref{fig:agg_multirounds} in Appendix~\ref{app: evaluation_extended}). The average increase summed across the 10 rounds is one defense unit. This means that subjects become less behavioral on average through learning. 


\subsection{Network (B) with Cross-over Edge}
This experiment used the attack graph from Figure \ref{fig:cross_over_edge_graph}. This attack graph is suitable to separate the spreading behavioral bias from the behavioral probability weighting, since for any $0 < \alpha \leq 1$, the optimal decision is to put zero defense units on the cross-over edge ($v_2,v_3$).
Figure \ref{fig:Avg_crit_blue} shows the average investment allocation on the cross-over edge based on 1450 investment decisions. We see that the proportion of subjects that are non-behavioral, i.e., invest nothing on the cross-over edge, is 29\%. We observe that the average of subjects' investments on the cross-over edge in each round, shows a weak downward trend (Figure~\ref{fig:agg_multirounds_blue} in Appendix~\ref{app: evaluation_extended}). Taken together, these human experiments provide support for our behavioral model with probability weighting and spreading factors.

\noindent \textbf{Generalizability of the study}:
The applicability of this subject study to security experts is motivated by the fact that numerous academic studies of even the most highly-trained specialists have shown that experts too have susceptibility to systematic failures of human cognition (e.g., \cite{haynes2012test,frechette2015handbook}). In the meta-review article~\cite{frechette2015handbook}, 9 of 13 studies that make a direct comparison between student and professional subject pools find no evidence of differing behavior, and only 1 out of 13 studies finds that professionals behave more consistently with theory. Moreover, recent research has shown that cybersecurity professionals’ probability perceptions are as susceptible to systematic biases as those of the general population~\cite{mersinas2015experimental, Mersinas2016AreTest}. Finally, even if security experts exhibit weaker biases, this can result in sub-optimal security investments and their effects may be magnified due to the magnitude of losses associated with compromised `real-world' assets. 
\section{Effect of Bias on Investments}\label{sec:theory}

In this section, we provide a simple example to illustrate the investment decisions by behavioral and non-behavioral defenders, and provide some intuition on why the optimal defense strategies under the two decision-making models differ.
In this example, we use the notion of a {\it min-cut} of the graph. Specifically, given two assets $s$ and $t$ in the graph, an edge-cut is a set of edges $\mathcal{E}_{c} \subset \mathcal{E}$ such that removing $\mathcal{E}_c$ from the graph also removes all paths from $s$ to $t$.  A min-cut is an edge-cut of smallest cardinality over all possible edge-cuts. As the example will show, the optimal investments by a non-behavioral defender (i.e., $\alpha = 1$) will generally concentrate the security investments on certain critical (i.e., min-cut) edges in the network.   
In contrast, behavioral defenders tend to spread their budgets throughout the network. 

Consider the attack graph shown in Figure \ref{fig:split_join_dependence_before}, with a single defender $D$ and a single target asset $v_5$ (with a loss of $L_5 = 1$ if successfully attacked).  Let the defender's budget be $B$, and let the probability of successful attack on each edge $(v_i, v_j)$ be given by $p_{i,j}(x_{i,j}) = e^{-x_{i,j}}$ (assuming $p_{i,j}^0 = 1$). This graph has two possible min-cuts, both of size $1$: the edge $(v_s, v_1)$, and the edge $(v_4, v_5)$.
The total loss function \eqref{eq:defender_utility} for the defender is given by
\begin{small}
\begin{align*}
C(\textit{x}) =  \max \left(e^{-(x_{s,1} + x_{1,2} + x_{2,4} + x_{4,5})}, e^{-(x_{s,1} + x_{1,3} + x_{3,4} + x_{4,5})}\right),
\end{align*}%
\end{small}
which reflects the two paths from the source $v_s$ to the target $v_t$. We note that the optimal solution of this constrained convex optimization problem satisfies the KKT conditions \cite{hillier2012introduction}. One can then verify (using KKT conditions \cite{hillier2012introduction}) that it is optimal for a non-behavioral defender to put all of her budget only on the min-cut edges, i.e., any solution satisfying $ x_{s,1} + x_{4,5} = B $ and $ x_{1,2}=x_{2,4}=x_{1,3}=x_{3,4}=0 $ is optimal.  The intuition of the above result is that from a non-behavioral defender's viewpoint, the probability of successful attack on any given path is a function of the sum of the security investments on the edges in that path. Thus, any set of investments on min-cut edges would be optimal since the sum of investments would be the whole security budget on each path of the graph. 

Now, consider a behavioral defender, i.e., a defender with $\alpha < 1$.  With the above expression for $p_{i,j}(x_{i,j})$ and using the Prelec function \eqref{eq:prelec}, we have
$w(p_{i,j}(x_{i,j})) = e^{-x_{i,j}^{\alpha}}$. 
Thus, the total (perceived) loss function \eqref{eq:defender_utility_edge} for a behavioral defender is 
\begin{align*}
C(x) =  \max \left(e^{-x_{s,1}^{\alpha} - x_{1,2}^{\alpha} - x_{2,4}^{\alpha} - x_{4,5}^{\alpha}},~ 
e^{-x_{s,1}^{\alpha} - x_{1,3}^{\alpha} - x_{3,4}^{\alpha} - x_{4,5}^{\alpha}}\right),
\end{align*}
which includes the two paths from the source $v_s$ to the target $v_5$. Again, one can verify (using the KKT conditions \cite{hillier2012introduction}) that the optimal investments are  
\vspace{-2mm}
\begin{align*}
x_{1,2} &= x_{2,4} = x_{1,3} = x_{3,4} = 2^{\frac{1}{\alpha-1}} x_{s,1} .\\
x_{s,1} &= x_{4,5} = \tfrac{B-4x_{1,2}}{2}= \tfrac{B}{2+4( 2^{\frac{1}{\alpha-1}}) }. 
\end{align*}

Comparing these two cases, 
the optimal investments of the non-behavioral defender yield a total loss of $ e^{-B}$, whereas the investments of the behavioral defender yield a total loss of $e^{-2^\frac{\alpha}{\alpha-1}} e^{-\frac{B}{1 +  2^{\frac{\alpha}{\alpha-1}}}}$, which is larger than that of the non-behavioral defender. 

\textbf{Interpretation:} The reason for this discrepancy can be seen by examining the Prelec probability weighting function in Figure \ref{fig:Prelec Probability weighting function}.  Specifically, when considering an undefended edge (i.e., whose probability of successful attack is $1$), the marginal reduction of the attack probability on that edge as {\it perceived} by a behavioral defender is much larger than the marginal reduction of true attack probability on that edge. Thus the behavioral defender is incentivized to invest some non-zero amount on that edge. Therefore, a behavioral defender splits her investments among the two non-critical sub-paths in the attack path. Note that the same insight holds for different baseline probabilities, but this shifting effect is greater when the slope of the behavioral probability weighting curve is higher (i.e., close to values of 1, 0, or where the cross-over happens between the behavioral curve and the diagonal).
A rational defender, on the other hand, correctly perceives the drop in probability, and thus prefers not to invest on the non-critical sub-paths
, instead placing her investment only on the critical edges $(v_s,v_1)$ or $(v_4,v_5)$ or both.

In the above example, we assumed all edges have the same sensitivity to investments. We provide the analysis for different edges' sensitivities 
in Appendix~\ref{app:sensitivity}.

\section{Learning Over Rounds}\label{sec:learning_rounds}

Here we consider a defender who plays multiple rounds of the game, learning from observing the attack in each round. In each round, each defender plays the single-shot game with the attacker, allocating all her security budget. She then uses information collected during this interaction to inform her future decisions. In particular, we consider two different forms of learning: (1) what can the defender learn about an attacker over time, and (2) how can repeated interactions lead to decrease in the defenders' extent of behavioral decision-making (i.e., increase in $\alpha$)? 
We answer these questions through casting them as repeated resource allocation and reinforcement learning problems, respectively. 
\subsection{Learning about the Attacker}\label{sec:learn-attacker}
Now, we assume that the defender can observe the attacker’s past actions, e.g., via an intrusion detection system~\cite{modelo2008determining} or user metrics~\cite{xie2010using}. 

We propose an algorithm through which the defender learns the attack paths over time, and distributes her investments optimally accordingly over the edges. 
In particular, the steps of this algorithm for this defense technique, as outlined in Algorithm~\ref{alg:attacker}, are as follows. First, for each round, we compute the empirical frequency of the attacker's actions over the past $N$ moves (i.e., the probability of choosing every attack path based on the most recent $N$ choices). Then, we compute the best response of the defender to a modified version of the cost $C_k(x_k)$: this is a weighted version of the cost where each path $P$ has a weight $\beta_P$ (computed from the previous step). The  complexity of the algorithm therefore depends on the number of attack paths. 

In Section~\ref{sec:evaluation}, we compare the investment decisions prescribed by Algorithm~\ref{alg:attacker} with those from our earlier single-shot setup where the defender exhibits no learning. In these comparisons, we consider three types of attackers: replay attackers, randomizing attackers, and adaptive attacker. Specifically, a \textit{replay attacker} chooses the same attack path for every critical asset in every round. 
Such behavior may be due to limited observations~\cite{alpcan2006intrusion}, or when the attack process is automated.
A \textit{randomizing attacker}, on the other hand, chooses an attack path (for every critical asset $v_m$) randomly each round, i.e., with probability following a uniform distribution over the possible attack paths in $\mathcal{P}_m$. 
Such attackers have also been studied in other work using attack graph models~\cite{wang2019attacking}. 
We consider a third attacker type, the \textit{adaptive attacker}, who chooses the least chosen attack path in the past $N$ moves (for every critical asset).

In contrast to replay attacker and randomizing attacker, we assume that the adaptive attacker is aware that the defender's strategy considers the most recent $N$ attacks, and thus the attacker engineers its attack history over a period of time so as to make additional gains on the future attack by choosing the least chosen attack path in the past $N$ moves. Note that the attacker does not have a budget, he just chooses an attack path to each critical asset.

\SetKwInput{KwInput}{Input} 
\SetKwInput{KwOutput}{Output}
\begin{algorithm}[t]
  \KwInput{Set of attack paths $\mathcal{P}_m$, number of rounds $N_{R}$ and history of attack paths $(P^{t-N},\cdots,P^{t-1})$}
  \KwOutput{Vector of investments over rounds, $\mathcal{O}$}
  Round Number = t = 0
  
  \While{$t < N_R$}
    {  \For {$v_m \in V_k$}
        {
            \For{Path $P \in \mathcal{P}_m$}
            {
              $ \beta^t_P = \frac{1}{N} \sum_{\tau = t - N}^{t-1} [P^{\tau} = P]_1 $ 
              
            }
        }
        $ C^t_{k}(x_k) = \displaystyle \sum_{v_{m} \in V_{k}} L_{m} \hspace{0.3mm} \Big( \hspace{0.3mm} \sum_{P \in P_{m}} \beta^t_P  \prod_{(v_{i},v_{j}) \in P} w(p_{i,j}(x_{i,j})) \hspace{0.3mm} \Big) $
        
        $x^{t}_{k} \in  {\displaystyle \textit{argmin}_{x_k \in X_{k}}} \hspace{0.5mm} C^t_{k}(x_k)$
        
        Append ($\mathcal{O}$, $x^{t}_k$) 
        
        }   
            
            Return $\mathcal{O}$
\caption{Learning Attack Paths}
\label{alg:attacker}
\end{algorithm} 

\subsection{Reinforcement Learning for Reducing Behavioral Decision-Making}\label{sec:improve-alpha}
As shown in Sections~\ref{sec:human-experiments} and~\ref{sec:theory}, the one-round investment decisions made by a behavioral defender $D_k$ based on the decision model in Equation \eqref{eq:defender_utility_edge} are sub-optimal. 
It is therefore of interest to understand whether such defender can reduce her behavioral biases in a multi-round defense game by using her experience from previous rounds. In this section, we propose a learning technique through which the defender can make such progress towards a more rational model, i.e., leads to $\alpha^j > \alpha^i$, for some $j > i$,  where $\alpha^i$ denotes the behavioral level in round $i$. 
Our proposed algorithm, outlined in Algorithm~\ref{alg:rl}, uses a reinforcement learning approach. 
Our algorithm is based on that of \cite{feltovich2000reinforcement}, adapted to our problem of security investment decision-making.

The algorithm proceeds as follows. Let
$q^t(\alpha_i)$ denotes the defender's propensity to invest according to the behavioral level $\alpha_i$ at round $t$.   We first initialize these propensities to the defender's initial behavioral level (i.e., $\alpha^0 = \alpha_i$, $q^0(\alpha_i) = A$, and $q^0(\alpha_j) = B, \forall j \neq i$).\footnote{In our evaluation, we show the convergence of Algorithm~\ref{alg:rl} under different possible values of the initial propensities of different behavioral levels (i.e., $A$ and $B$ in Algorithm~\ref{alg:rl}). We also show in Appendix~\ref{app:RL_convergence_analysis} that the convergence of Algorithm~\ref{alg:rl} depends on the true total loss of the investment, not the initial propensities. 
} Then, for every round $t$, the defender does not know her behavioral level but she draws her defense budget decision in accordance to her reinforcement level. After the defender distributes her defense budget, she receives corresponding reinforcement $R^t$ (which is the difference between the true loss $\hat{C}^t(x^t_k)$ calculated with the investments (budget allocation on edges) in round $t$, denoted by $x^t_k$, and the maximum possible true loss $\hat{C}_{max}$ (which is the initial loss). Thus, if the defender invests according to a more rational behavior (i.e., higher $\alpha$) in round $t$, she receives higher reinforcement and thus the propensity to choose this investment again in next rounds ($q^{t+1}(\alpha_i)$) increases. For all other investments that are not observed in this round, the propensities of the corresponding behavioral levels do not change. Then, we update the probability distribution for the investments (resp. behavioral levels) for the next round. We repeat the process until we reach convergence (where the reinforcement learning model chooses  $\alpha_i = 1$ with a probability sufficiently close to 1) or we reach the maximum number of rounds $N_R$. The output of our algorithm is a time-series of behavioral level values. We emphasize that the learning comes from the reinforcements received each round which controls the propensity of  the defender to choose particular budget distributions in next rounds and that the defender does not know the optimal investments apriori.

\textbf{Convergence of Algorithm~\ref{alg:rl} to rational behaviour:} 
The convergence of Algorithm~\ref{alg:rl} depends on the relation between the total loss (true cost) under rational behavior $\alpha = 1$ and the total loss (true cost) under bias $\alpha < 1$. 
In the interest of space, we state this result in Lemma~\ref{lemma: RL_convergence} and provide its proof in Appendix~\ref{app:RL_convergence_analysis}.

\SetKwInput{KwInput}{Input}
\SetKwInput{KwOutput}{Output}
\begin{algorithm}
  \KwInput{Set of behavioral levels $\mathcal{\alpha}$ and number of rounds $N_{R}$}
  \KwOutput{Vector of behavioral level over rounds $\mathcal{O}$}
  Round Number = t = 0
 
            $q^{0}(\alpha_i) =  A$ and $q^{0}(\alpha_j) =  B  \forall j \neq i$ 
 
  \While{$t < N_R$ or \text{not Convergence to $\alpha_i = 1$}}
      {  \For {$\alpha_i \in {\mathcal{\alpha}}$}
        {
        \If  {$\alpha_i$ \text{was observed in round t}}
           {
            $x^{t}_{k} \in  {\displaystyle \textit{argmin}_{x_k \in X_{k}}} \hspace{0.5mm} C^t_{k}(x_k,\alpha_i)$
            
            
            $R^t = \hat{C}_{max} - \hat{C}^t_{k}(x^{t}_{k})$

            $q^{t+1}(\alpha_i)$ = $q^{t}(\alpha_i) + R^t$  
            }
        \Else
            { 
              $q^{t+1}(\alpha_i)$ = $q^{t}(\alpha_i)$ 
            }
        
    
              $p^{t+1}(\alpha_i)$ = $\frac{q^{t+1}(\alpha_i)}{\sum_{\alpha_i \in \mathcal{\alpha}} q^{t+1}(\alpha_i)}$
        }
        
         Sample random $\alpha_i$
            with probability  $p^{t+1}(\alpha_i)$ to get $\alpha^{t+1}$
            

        Append ($\mathcal{O}$, $\alpha^{t+1}$) 
        
        }   
            
            Return $\mathcal{O}$
\caption{Reinforcement Learning to Reduce Behavioral Biases}
\label{alg:rl}
\end{algorithm} 

\subsection{Hybrid-Learning Algorithm}\label{sec:hybrid-learning}
In Algorithm~\ref{alg:rl} the defender learns through observing her payoffs in the last recent rounds. In Algorithm~\ref{alg:attacker}, the defender learns the attacker's chosen paths. Here, we combine these two forms of learning to create a hybrid learning algorithm. This algorithm is a modified version of Algorithm~\ref{alg:rl} where the cost $C^t_k(x_k, \alpha_i)$ is the cost proposed in Algorithm~\ref{alg:attacker}, which changes each round as the defender updates the weights of each path according to the history of attack paths. We will evaluate this hybrid-learning algorithm in Section~\ref{sec:evaluation} and will compare it with both of Algorithm~\ref{alg:attacker} and Algorithm~\ref{alg:rl}, described earlier in this section. 
\section{Evaluation}\label{sec:evaluation}

\begin{figure*}[t] 
\begin{subfigure}[t]{.33\textwidth}
 \centering
    \includegraphics[width=\linewidth]{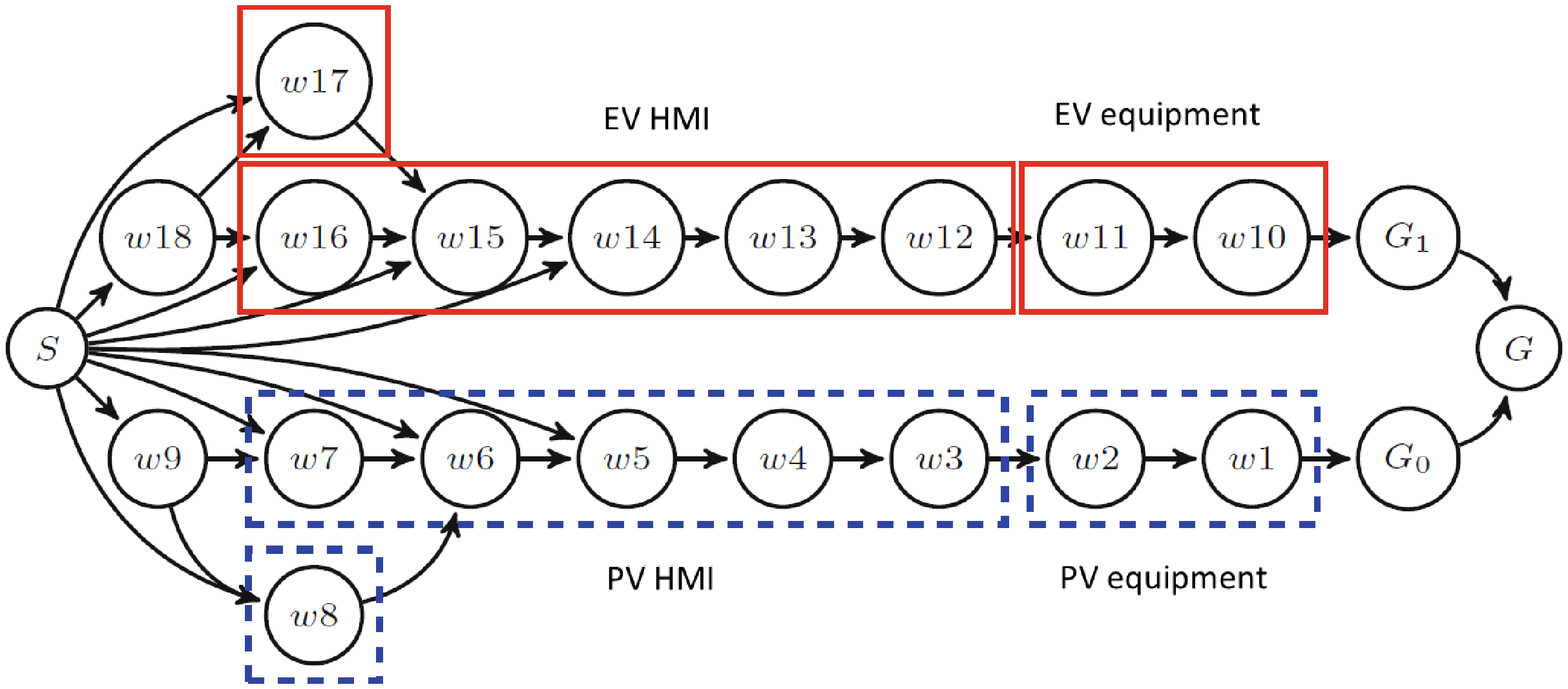}
\caption{Attack Graph of DER System}
\label{fig:DER_Attack_Graph} 
\end{subfigure} 
 \begin{subfigure}[t]{.33\textwidth}
\centering
  \includegraphics[width=\linewidth,height=22mm]{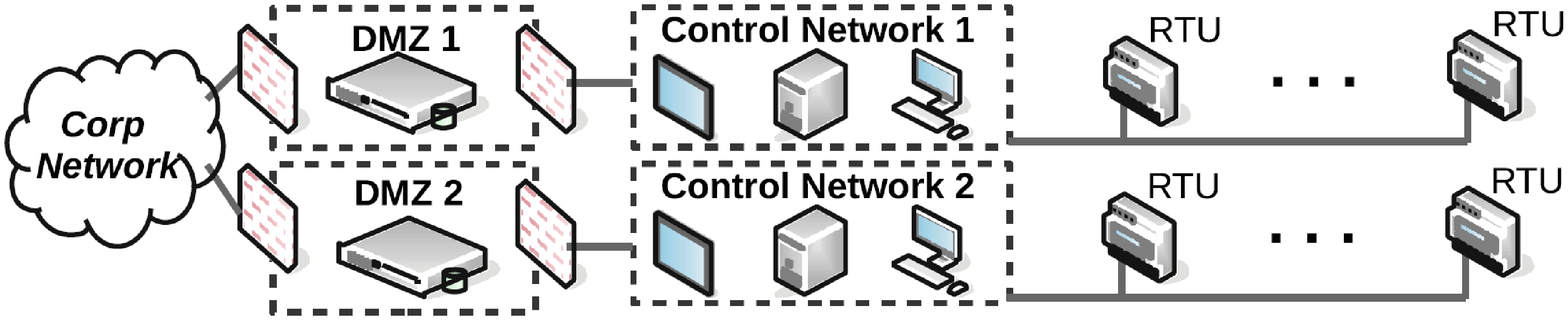}
  \caption{A high level network overview of the SCADA system}
 \label{fig:Scada_High_level_overview}
\end{subfigure}
 \begin{subfigure}[t]{.33\textwidth}
\centering
\includegraphics[width=\linewidth]{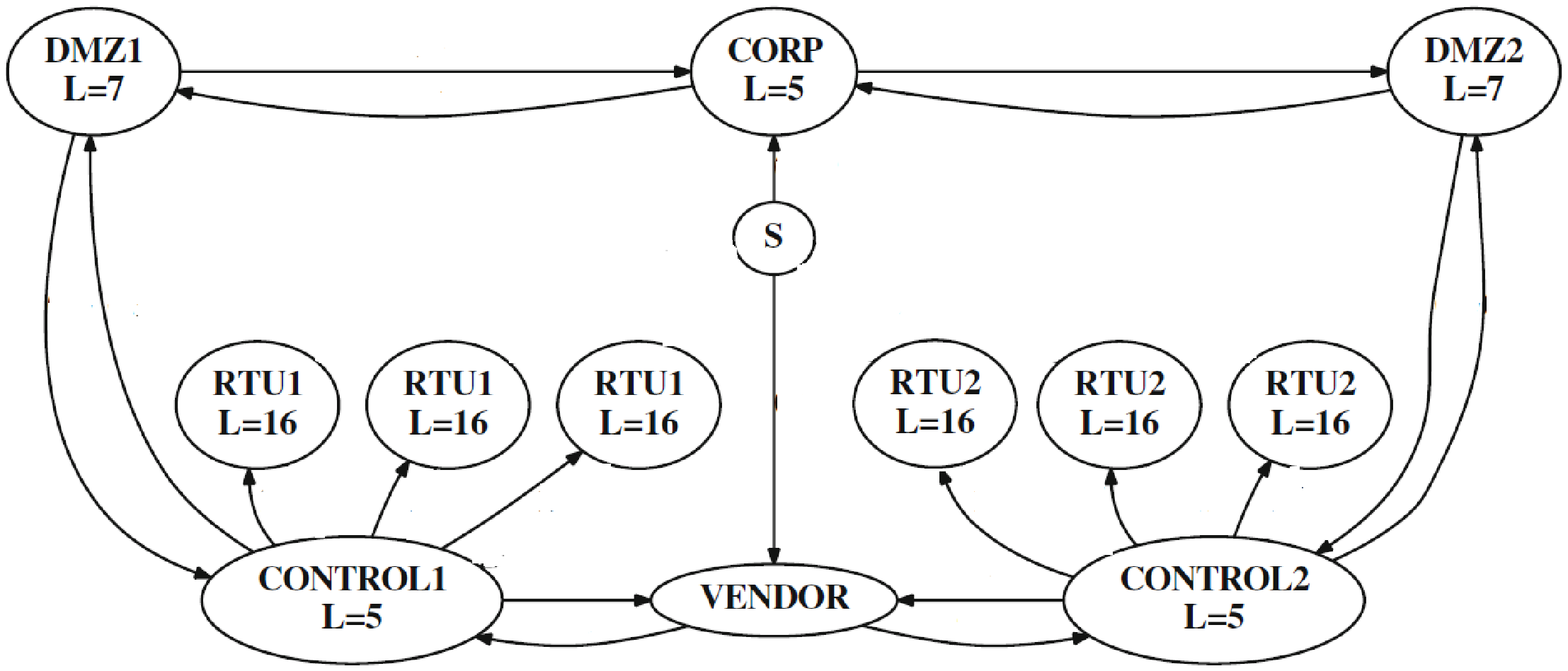}
\caption{Attack Graph of SCADA System}
\label{fig:SCADA_Attack_Graph}
\end{subfigure}
 \caption{Attack graphs of DER.1 and SCADA case studies. The attack graphs of the remaining systems are given in Appendix~\ref{app:systems_overview}.}
\label{fig:Attack_Graphs}
 \end{figure*}

Our evaluation of \name aims to answer the following questions:
\begin{itemize}
\item What is the gain of using \name for guiding behavioral decision-makers towards rational decision-making?
\item How can we decrease level of behavioral bias over rounds?
\item How does each system parameter affect the overall security level of the system with behavioral decision-making?
\end{itemize}

\subsection{Experimental Setup}
\textbf{Dataset Description:} We use five synthesized attack graphs that
represent real-world interdependent systems with different sizes to evaluate our setups, i.e., different attacks, defense, and learning (See Table~\ref{tbl:gain_morshed}). Specifically, we consider 5 popular interdependent systems from the literature which are: DER.1~\cite{jauhar2015model}, SCADA (with internal attacks)~\cite{hota2016optimal}, SCADA (with only external attacks), IEEE 300-bus smart grid~\cite{khanabadi2012optimal}, E-commerce~\cite{modelo2008determining}, and VOIP~\cite{modelo2008determining}. In all of these systems, nodes represent attack steps (e.g., taking privilege of control unit software in SCADA,  accessing customer confidential data such as credit card information in E-commerce). Now, we give a detailed explanation of one of these systems; the SCADA system (see Appendix~\ref{app:systems_overview} and \cite{jauhar2015model, modelo2008determining, khanabadi2012optimal} for detailed description of the rest of the systems). We generate the attack graphs using the CyberSage tool~\cite{jauhar2015model} which maps the failure scenarios of the system automatically into an attack graph given the workflow of that system, the security goals, and the attacker model.

\begin{figure*}[t] 
\centering
\begin{subfigure}[t]{.32\textwidth}
\centering
 \includegraphics[width=\linewidth]{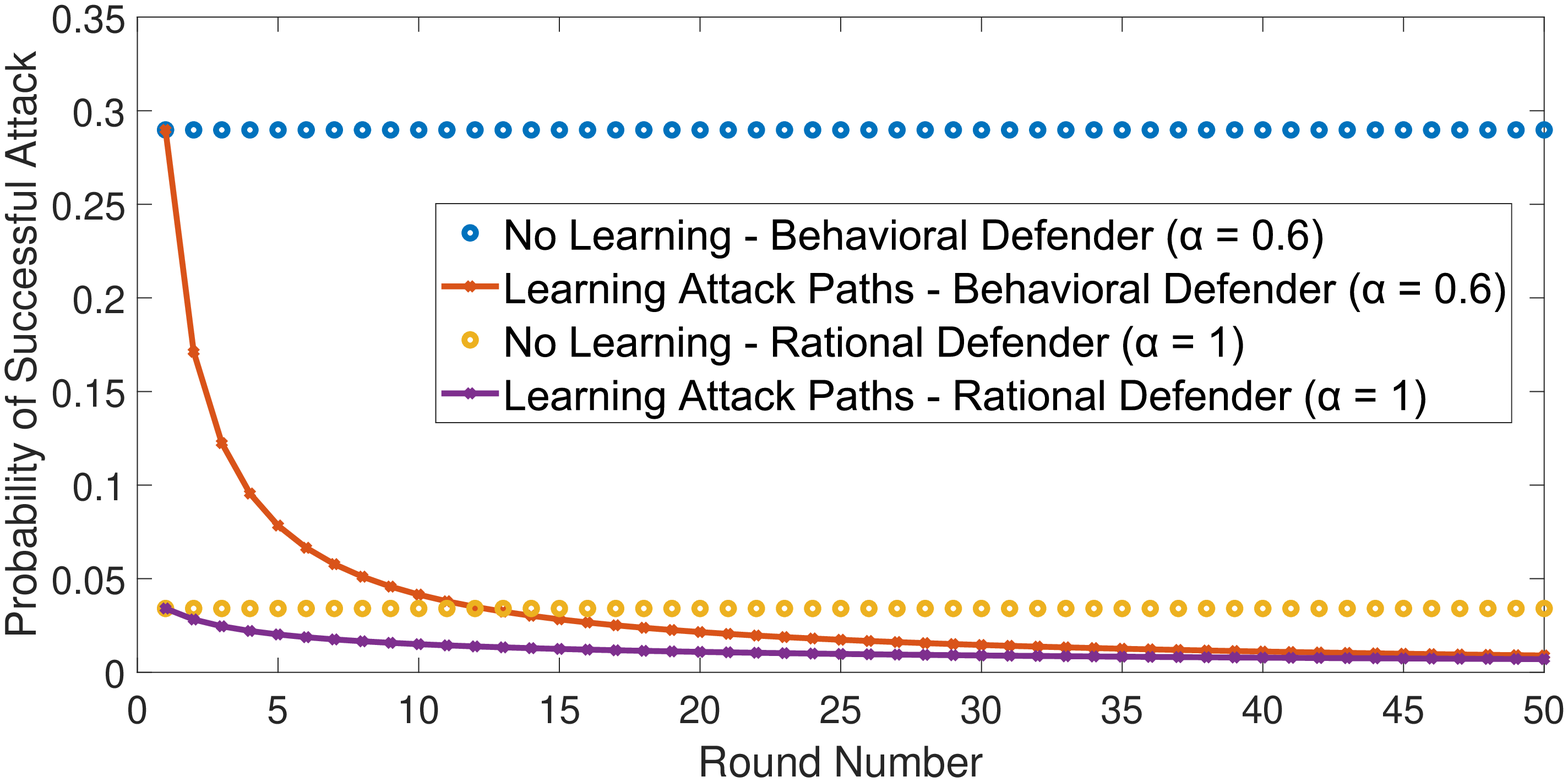}
  \caption{Attacker chooses same attack paths}
  \label{fig:learning_attack_paths_scada}
 \end{subfigure} 
 \begin{subfigure}[t]{.32\textwidth}
\centering
  \includegraphics[width=\linewidth]{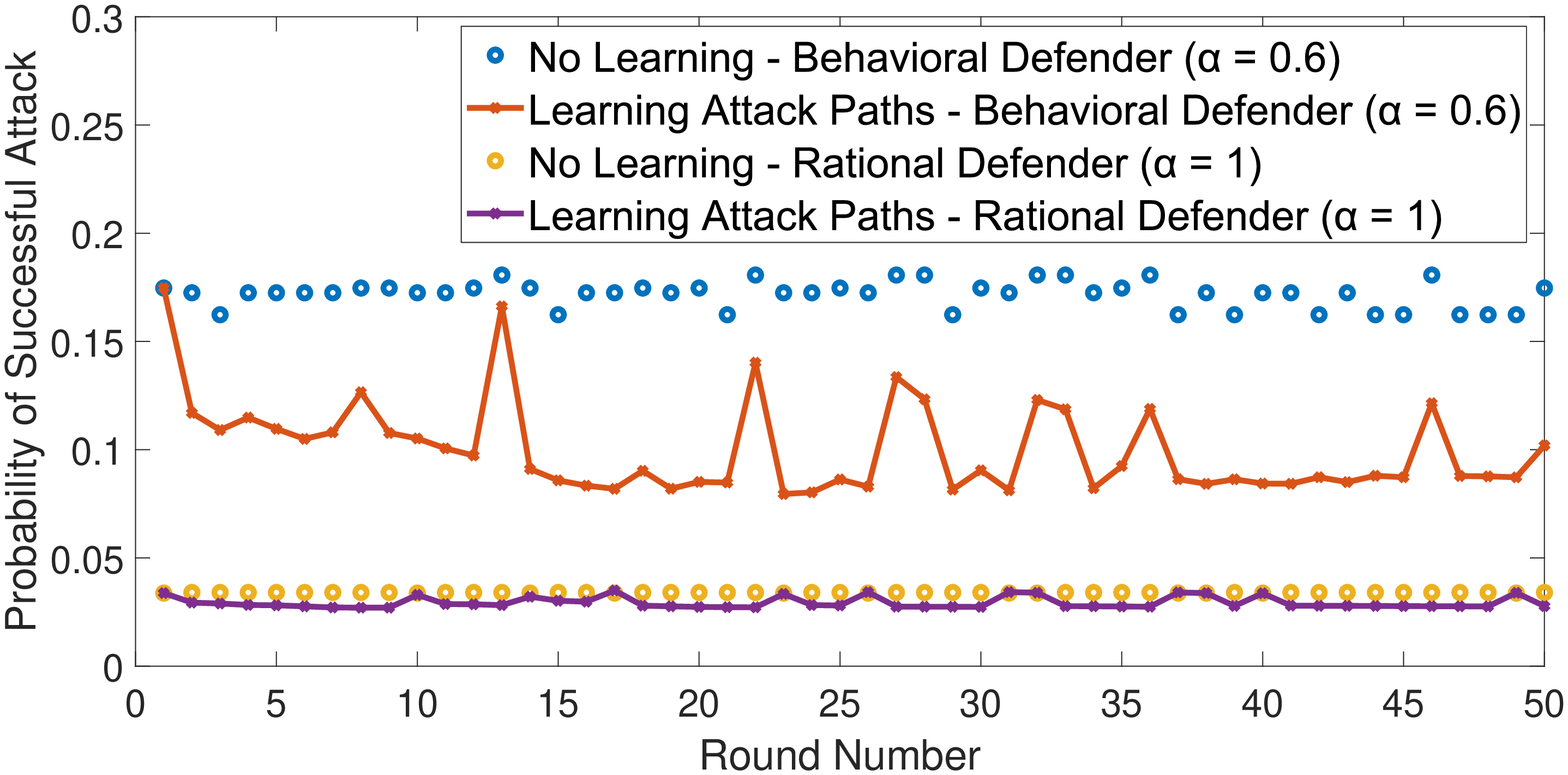}
  \caption{Attacker chooses attack paths randomly}
  \label{fig:learning_attack_paths_scada_random}
 \end{subfigure}
  \begin{subfigure}[t]{.32\textwidth}
\centering
  \includegraphics[width=0.95\linewidth]{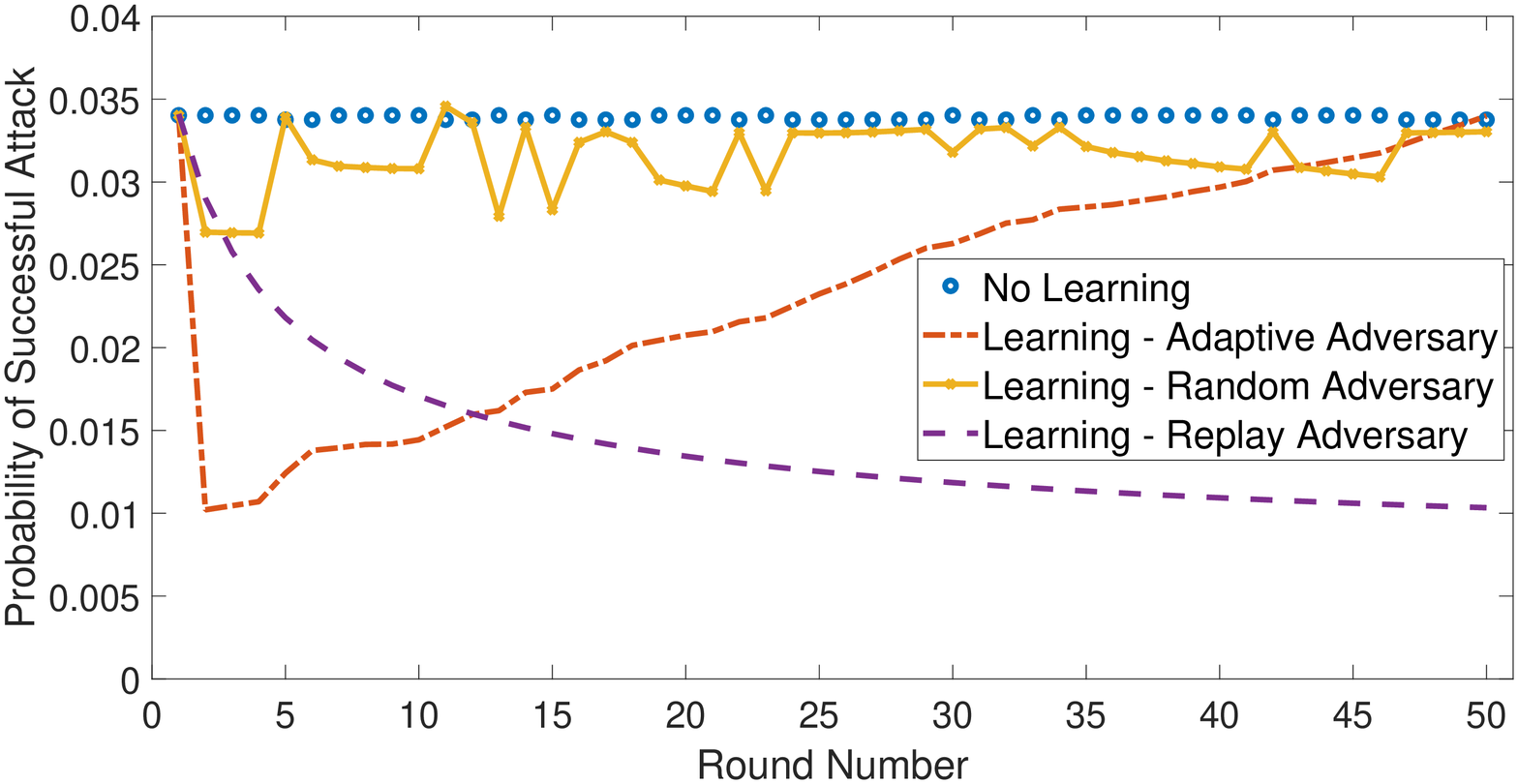}
  \caption{Different attack types comparison}
  \label{fig:learning_attack_paths_scada_all}
 \end{subfigure}
 \caption{The effect of learning attack paths over the rounds. The learning is useful for both behavioral and rational defenders. Moreover, behavioral defender with learning attack paths can eventually reach same security level as rational defender (specifically if the attacker chooses same attack path for each critical asset over rounds). The adaptive attacker is the most challenging attack type.}
\label{fig:learning_rounds}
 \end{figure*}
\begin{table}[t]
\caption {The one-round gain of \name compared to behavioral investment decisions for the five studied interdependent systems.}
\label{tbl:gain_morshed}
\centering
\resizebox{0.9\columnwidth}{!}
{%
\begin{tabular}{|l|l|l|l|l|l|l|}
\hline
\multicolumn{1}{|l|}{\text{\bf System}}
& \multicolumn{1}{l|}{\bf $\#$ Nodes}
& \multicolumn{1}{l|}{\bf $\#$ Edges}
& \multicolumn{1}{l|}{\bf \shortstack{ $\#$ Min-cut Edges}}
& \multicolumn{1}{l|}{\bf \shortstack{ $\#$ Critical Assets}}
& \multicolumn{1}{l|}{\bf Avg Gain}
& \multicolumn{1}{l|}{\bf Max Gain}
\\
\cline{1-5}
\hline
SCADA-external & 13 & 20 & 2 & 6 & 1.43 & 2.63\\
\hline
SCADA-internal \cite{hota2016optimal} & 13 & 26 & 8 & 6 & 4.43 & 9.42 \\
\hline
DER.1 \cite{jauhar2015model} & 22 & 32 & 2 & 2 & 1.29 & 2.38 \\  
\hline
E-Commerce \cite{modelo2008determining} & 18  & 26 & 1 & 4 & 3.70 & 18.28 \\
\hline
VOIP \cite{modelo2008determining} & 20 & 28  & 2 & 4 & 4.46 & 18.66\\ 
\hline
IEEE 300-bus \cite{khanabadi2012optimal} & 300  &  822  & 98 & 69 & 5.85 & 11.25\\
\hline
\end{tabular}%
}%
\vspace{-0.1in}
\end{table}

\textbf{SCADA system description}:\label{app: scada-system-description}
The SCADA system (shown in Figure~\ref{fig:Scada_High_level_overview}) is composed of two control subsystems, where each incorporates a number of cyber components, such as control subnetworks and remote terminal units (RTUs), and physical components, such as, valves controlled by the RTUs. 
This system is architected following the NIST guidelines for industrial control systems. For example, each subsystem is separated from external networks through a demilitarized zone (DMZ). The purpose of a DMZ is to add an additional layer of security between the local area networks of each control subsystem and the external/corporate networks, from where external attackers may attempt to compromise the system. 
The system implements firewalls both between the DMZ and the external networks, as well as between the DMZ and its control subnetwork. Therefore, an adversary must bypass two different levels of security to gain access to the control subnetworks.

Mapping this system to our proposed security game model, each control subnetwork is owned by a different defender. These two subsystems are interdependent via the shared corporate network, as well as due to having a common vendor for their control equipment. The resulting interdependencies map to the attack graph shown in Figure~\ref{fig:SCADA_Attack_Graph}. 
The ``Corp'' and the ``Vendor'' nodes connect the two subnetworks belonging to the two different defenders and can be used as jump points to spread an attack from one control subsystem to the other. This system has six critical assets (i.e., 3 RTUs, Control Unit, CORP, and DMZ). The compromise of a control network ``CONTROL $i$'' will lead to loss of control of all 3 connected RTUs.
Now, we present the various system parameters.

\noindent \textbf{Baseline Probability of successful attack:} Each edge in the attack graphs represents a real vulnerability. To create the baseline probability of attack on each edge (i.e., without any security investment), we first create a table of CVE-IDs (based on real vulnerabilities reported in the CVE database for 2000-2019). We then followed \cite{homer2013aggregating} to convert the attack's metrics (i.e., attack vector (AV), attack complexity (AC)) to a baseline probability of successful attack (e.g., Table \ref{tbl:cvss_cve_der_scada} in Appendix illustrates such process for SCADA and DER.1). Interestingly, we show that the gain of rational vs. behavioral investments exists for any combination of baseline probabilities (as will be shown in Section~\ref{sec: eval_multiple_def}).


\noindent \textbf{Security Budget:} We assume that the total budget available at the defenders' organization is $B$, and that an amount $BT$ of this budget is set aside for security investments. We refer to $BT < 0.3B$, $0.3B < BT < 0.6B$, and $BT > 0.6B$, as low, medium, and high security budgets. For instance, $ BT = \$10$ and $\$20 $ reflect low and moderate budgets, respectively  and $ BT \geq \$30$ reflects high budgets in SCADA system given that $B = \$50$.
We emphasize that the gain of our proposed techniques exists for any choice of budget (as will be shown in Section~\ref{sec: eval_multiple_def}).

\noindent \textbf{Convergence to Optimal Solution:} In our experiments, to find the optimal investments, we use the notion of {\it best response dynamics}, where the investments of each defender $D_k$ are iteratively updated based on the investments of the other defenders. In each iteration, the optimal investments for defender $D_k$ can be calculated by solving the convex optimization problem in \eqref{eq:defender_utility_edge}.\footnote{
Note that in the results of learning attack paths and Hybrid learning techniques, we use different cost function (shown in Algorithm~\ref{alg:attacker}).} Note that the best response dynamics converge to a Nash equilibrium~\cite{hota2016optimal} and we study the security outcomes at that equilibrium.


\subsection{Gain from Using \name in One Round}
Here, we show the gain that behavioral security decision-maker would have using \name.

\noindent \textbf{Reduction in Defender's Total Loss:} To show the gain of our proposed algorithm, we quantitatively compare the total system loss of the aforementioned five systems in two scenarios which are assuming behavioral decision-maker without the help of \name and with the help of \name investments, respectively. We then calculate the gain as the ratio of the total system loss by behavioral decision-maker to the total system loss by \name to quantify the benefit of using our proposed algorithm.

\textbf{1) Average Gain:} We define the Average Gain as the ratio of the weighted sum of total system loss by behavioral decision-maker to the total system loss by \name assuming that 50\% of the decision-makers are fully rational (with $\alpha = 1$) and 50\% are behavioral defenders ($\alpha \in [0.4,1)$); this is consistent with the range of behavioral parameters from prior experimental studies \cite{gonzalez1999shape} and our subject study. Average Gain for all systems is shown in Table~\ref{tbl:gain_morshed}. 

\textbf{2) Maximum Gain:} We define the Maximum Gain as the ratio of the total system loss by the highest behavioral defender ($\alpha = 0.4$) to the total system loss by rational ($\alpha = 1$) decision-maker (computed by \name). Table~\ref{tbl:gain_morshed} shows maximum gains which are 2.38, 9.43, 2.63, 11.25, 18.28, and 18.66 for the  DER.1, SCADA-internal, SCADA-external, IEEE 300-bus, E-commerce, and VOIP respectively.  

\begin{figure*}[t] 
\centering
\begin{subfigure}[t]{.32\textwidth}
 \centering
\includegraphics[width=0.89\linewidth]{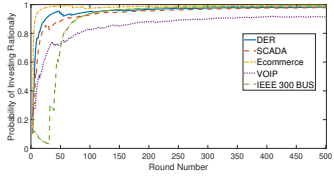}
  \caption{Convergence of Reinforcement learning to rational behavior for the five studied interdependent systems.} 
  \label{fig:RL_all_systems}
\end{subfigure}\hfill
\begin{subfigure}[t]{.32\textwidth}
 \centering
   \includegraphics[width=0.89\linewidth]{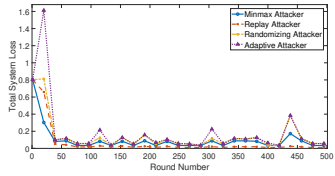}
  \caption{Defense Enhancement under Hybrid Learning for different attack types. The spikes (that represents investing suboptimally) decreases in later rounds.} 
  \label{fig:Hybrid_learning_costs}
 \end{subfigure}\hfill 
 \begin{subfigure}[t]{.32\textwidth}
 \centering
   \includegraphics[width=0.9\linewidth]{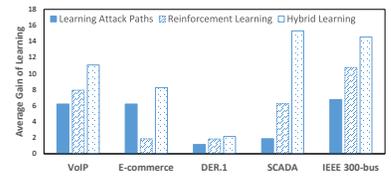}
  \caption{Average Gain in Total System Loss for the different Learning techniques. The Hybrid Learning is superior for all five systems.} 
  \label{fig:comparison_different_learning_techniques}
 \end{subfigure}
 \caption{(a) shows the convergence of Reinforcement learning for all systems. (b) shows the effect of Hybrid learning for each attack type. In (c), we show the average gain of learning for all systems.}
\label{fig:learning_combined_systems_effects}
 \end{figure*}

\vspace{-3mm}
\subsection{Learning over Rounds Results}
Now, we consider the different setups where the defender learn over rounds using our proposed algorithms in Section~\ref{sec:learning_rounds}. For some results, we only show results on the SCADA attack graph as we observe similar patterns on the remaining studied attack graphs. 

\textbf{1) Learning of attack paths:}
We show the effect of learning attack paths over the rounds for all of the possible attack scenarios described in Section~\ref{sec:learning_rounds}. We consider the five systems described earlier and simulate our learning algorithm over 50 rounds with considering medium budget. For each round, the attacker chooses one path for compromising each critical asset (for SCADA, we have six critical assets (i.e., 3 RTUs, Control Unit, CORP, and DMZ) and thus each round the attacker chooses six paths, one for each critical asset) and then the overall probability of successful attack is calculated. We show that the learning of attack paths is useful for both behavioral and rational defenders. Specifically, Figure~\ref{fig:learning_attack_paths_scada} shows such effect of learning if the attacker chooses same attack paths for each critical asset. Also, Figure~\ref{fig:learning_attack_paths_scada_random} shows that our proposed algorithm helps enhancing system security even if the attacker chooses attack paths randomly over rounds since it captures an approximate distribution of the attacker choice of the paths over the rounds.  Interestingly, behavioral defender that learns attack paths can eventually reach comparable security level as rational defender (with same security level if the attacker chooses same attack path for each critical asset over rounds; here, after 40 rounds as shown in Figure~\ref{fig:learning_attack_paths_scada}). Moreover, we compare the learning effect for all attack types, defined in Section~\ref{sec:learning_rounds}, in Figure~\ref{fig:learning_attack_paths_scada_all} which shows that adaptive attacker is the most challenging attack type.

\textbf{2) Reinforcement learning of behavioral level:} Now, we show the performance of our reinforcement learning algorithm to guide behavioral decision-makers to rational behavior. Here, we consider the attacker who chooses the most vulnerable path to each target asset introduced in Section~\ref{sec:model}. For each system of the five systems, we run our learning algorithm over 500 rounds with considering medium budget. For each round, the attacker chooses the most vulnerable path for compromising each critical asset. First, Figure~\ref{fig:RL_all_systems} shows the convergence of our algorithm over the rounds to rational behavior (i.e., $\alpha$ = 1) for all of the five systems where the probability of having rational behavior after learning over 100 rounds is more than 0.9 and approaches 1 by the end of 500 rounds (for four systems from the five systems). Note that here we show the convergence when the initial behavior was $\alpha = 0.2$ (i.e., $\alpha^{0} = 0.2$). Such convergence would happen for any behavioral defender (with any $\alpha < 1$) given enough learning. This also shows that behavioral defender with our proposed Reinforcement learning algorithm can eventually reach optimal investment decisions (that leads to comparable security level as rational defender). 

Figure~\ref{fig:RL_all_systems} shows the rate of convergence of our Reinforcement-learning algorithm for the five case studies. It worth noting that learning is slower for VOIP compared to the other four systems. The reason is the higher criss-cross edges across the VOIP system (see Figure~\ref{fig:voip_High_level_overview} in Appendix~\ref{app:systems_overview}). This also sheds the light that each system has its own characteristics and may need further parameter tuning for enhancing convergence.

\textbf{Initial values of propensities:} Recall from Algorithm~\ref{alg:rl} that $A$ and $B$ represent the initial propensities for investing with the initial behavioral level and the propensities of other possible behavioral levels, respectively. To test convergence under different setups, we iterate $A$ over the values $\{0.1, 0.2,0.6,1,1.6,2\}$ while keeping $B = 0.1$ to simulate different propensities for investing with initial behavior level $\alpha_0$. In all of the experiments, the algorithm converges to rational behaviour with an average of $400$ iterations (i.e., $P^{t}(1) > 0.95$ at $t \geq 400$).  

\textbf{3) Hybrid-learning Results:}
Here, we show the performance of our proposed Hybrid-learning Algorithm in Section~\ref{sec:learning_rounds}. Figure~\ref{fig:Hybrid_learning_costs} shows the enhancement of defense (represented by total system loss) over rounds under the Hybrid-learning for all proposed attack types. Note that we let the initial attack to be the same for all of the four attack types. We note that our proposed Hybrid-learning algorithm is effective in reducing total system loss for all attack types with emphasizing that it is more effective with the replay attacker (i.e., the attacker that chooses same attack path for each critical asset every round) and the minmax attacker (i.e., the attacker that chooses the attack path with the highest probability for each critical asset every round). The enhancement is also noticeable for the two other attack types (i.e., randomizing and adaptive) but with less magnitude since capturing the attack patterns by the defender is more challenging in these two attack types. Note also that the spikes in the figure corresponds to the rounds in which the defender invest sub-optimally. These spikes decrease with rounds since the probability of investing behaviorally decreases as defender learns and enhance her budget distribution over rounds. 

\noindent \textbf{Benefit of Hybrid Learning:} We show the benefit of the learning techniques by calculating the Average Gain of Learning (which is the ratio of total system loss after learning to the total system loss with no learning averaged over the four attack types we study in this paper). Figure~\ref{fig:comparison_different_learning_techniques} shows the Average Gain of learning for the three earning techniques: Learning attack paths only (Algorithm~\ref{alg:attacker}), Reinforcement Learning only (Algorithm~\ref{alg:rl}), and Hybrid Learning. We observe the superiority of Hybrid Learning compared to using only one of the two learning techniques for all of the five systems. The intuition is that this Hybrid learning combines both learning behavioral level with learning attack paths.
\subsection{Baseline Systems} We compare \name with two baseline systems: the seminal work of \cite{sheyner2002automated} for security investment with attack graphs on attack graph generation and investment decision analysis\footnote{More recent approaches (e.g., \cite{zhang2016network}) follow the same strategy proposed in \cite{sheyner2002automated}.} and \cite{lippmann2006validating} for placing security resources using defense in depth technique which traverses all edges  that can be used to compromise each critical asset and distribute resources equally on them. 
In \cite{sheyner2002automated}, the defense mechanism is to select the minimal set $C$ of edges that, if removed from the attacker's arsenal, will prevent her from reaching the target asset (there can be multiple sets in case of non-uniqueness). This is equivalent to our min edge-cut. We compare~\cite{sheyner2002automated} and ~\cite{lippmann2006validating} with \name under both single and  multi-round setups. We compare the two methods in Table~\ref{tbl:Baseline_Comp_Table} by calculating the probability of successful attack (PSA) and show the superiority of \name in multi-round for all different attack types. Note that the defense investments given by \name for non-behavioral defenders is identical to that determined by \cite{sheyner2002automated} in single-round setup.

\subsection{Evaluation of Multiple-defender Setups}\label{sec: eval_multiple_def}
Here, we evaluate our proposed algorithm in multiple-defender setups. There are six parameters that could affect the total loss of the defender. The six parameters are: defenders' security budget availability (Low, Moderate, and High), the defense mechanism (Individual, and  Joint), the budget distribution among defenders (Symmetric, and Asymmetric), the degree of interdependency (number of edges between defenders' subnetworks), the sensitivity of edges to investments (the hyperparameter $s_{i,j}$), and the edges' baseline probabilities of successful attacks (the hyperparameter $p^{0}_{i,j}$). When studying the impact of a specific parameter, we fix the remaining parameters to their default values. Next, we study the impact of each system parameter with the behavioral decision-making and identify the effects of these system parameters on the degree of suboptimality of security outcomes due to behavioral decision-making in the two-defenders SCADA system.  
\renewcommand{\arraystretch}{0.6}
\begin{table}
\caption {Comparison of \name and baseline systems for different attacks scenarios. We consider here rational defender for \name. The column ``System Setup" shows the specific scenario; the second, third, and forth columns show the respective probability of successful attack (PSA) under \cite{sheyner2002automated},  \cite{lippmann2006validating}, and \name for each scenario.} 
\label{tbl:Baseline_Comp_Table}
\centering
\resizebox{0.75\columnwidth}{!}{%
\begin{tabular}{|l|l|l|l|l|l|l}
\hline
\multicolumn{1}{|l|}{\text{System Setup}}
& \multicolumn{1}{l|}{\cite{sheyner2002automated} }
& \multicolumn{1}{l|}{\cite{lippmann2006validating}}
& \multicolumn{1}{l|}{\name}\\
\hline
\multicolumn{4}{|c|}{\textbf{DER.1}} \\
\hline
\multicolumn{1}{|l|}{ }
& \multicolumn{3}{c|}{PSA}\\
\hline
Single-round   &  \textbf{0.075}  &  0.208  &  \textbf{0.075}  \\
\hline
Multi-round, Random Att.   &  0.095  &  0.205 & \textbf{0.080}  \\
\hline
Multi-round, Replay Att.   & 0.075  &  0.208   & \textbf{0.037}  \\
\hline
Multi-round, Adaptive Att.  & 0.091  & 0.209  &  \textbf{0.080} \\
\hline
\multicolumn{4}{|c|}{\textbf{SCADA}} \\
\hline
Single-round   & \textbf{0.035}  & 0.110 & \textbf{0.035}   \\
\hline
Multi-round, Random Att.   & 0.034  &  0.582 & \textbf{0.029} \\
\hline
Multi-round, Replay Att.   & 0.033  &  0.110 & \textbf{0.010}  \\
\hline
Multi-round, Adaptive Att.  & \textbf{0.035}  &  0.582 & \textbf{0.035}  \\
\hline
\multicolumn{4}{|c|}{\textbf{VOIP}} \\
\hline
Single-round  & \textbf{0.337}  &  0.556 & \textbf {0.337}\\
\hline
Multi-round, Random Att.  & 0.348  &  0.559 & \textbf {0.313} \\
\hline
Multi-round, Replay Att.   & 0.337  &  0.556 & \textbf{0.084}    \\
\hline
Multi-round, Adaptive Att.  & 0.354 & 0.559 & \textbf{0.313}  \\
\hline
\multicolumn{4}{|c|}{\textbf{E-commerce}} \\
\hline
Single-round   & \textbf{0.124}  &  0.276 & \textbf{0.124}  \\
\hline
Multi-round, Random Att.   & 0.139  &  0.572 & \textbf{0.097} \\
\hline
Multi-round, Replay Att.   & 0.124  &  0.276 & \textbf{0.007}    \\
\hline
Multi-round, Adaptive Att.  & 0.139 &  0.569 & \textbf{0.097}   \\
\hline
\multicolumn{4}{|c|}{\textbf{IEEE 300-BUS}} \\
\hline
Single-round   & \textbf{0.431}  & 0.653 &  \textbf{0.431}  \\
\hline
Multi-round, Random Att.   & 0.439 & 0.680 & \textbf{0.168}  \\
\hline
Multi-round, Replay Att.   & 0.431  &  0.653 & \textbf{0.086}  \\
\hline
Multi-round, Adaptive Att.  & 0.448 & 0.680 &  \textbf{0.186}  \\
\hline
\end{tabular}}
\vspace{-3mm}
\end{table}

\begin{figure}[t]
\begin{minipage}[t]{.48\linewidth}
\centering
  \includegraphics[width=\linewidth]{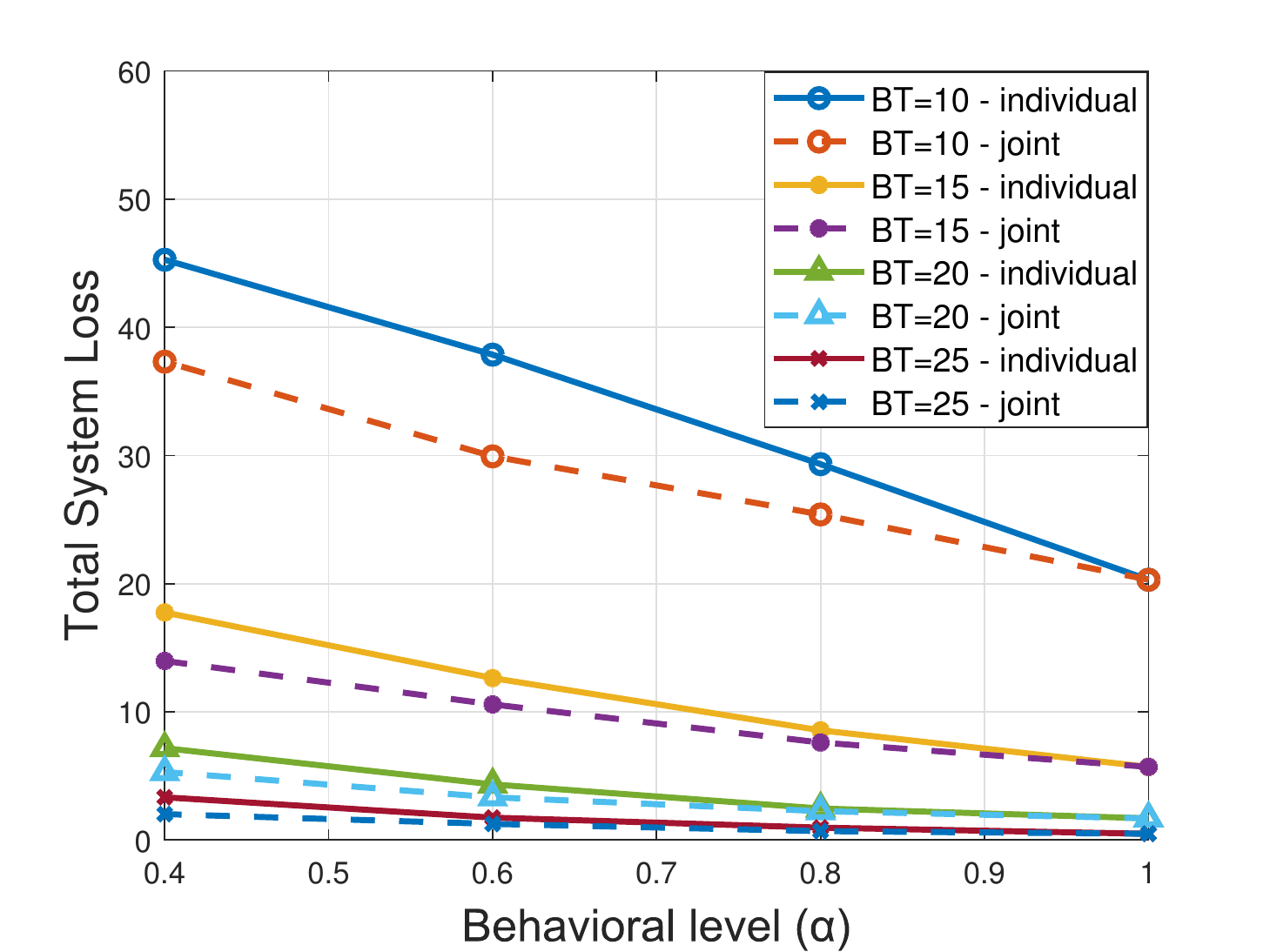}
  \caption{Comparison between individual and joint defense mechanisms. Joint defense is superior under asymmetric budget distribution.} 
\label{fig:Defense_mechanisms}
\end{minipage}\hfill 
\begin{minipage}[t]{0.48\linewidth}
 \centering
 \includegraphics[width=\linewidth]{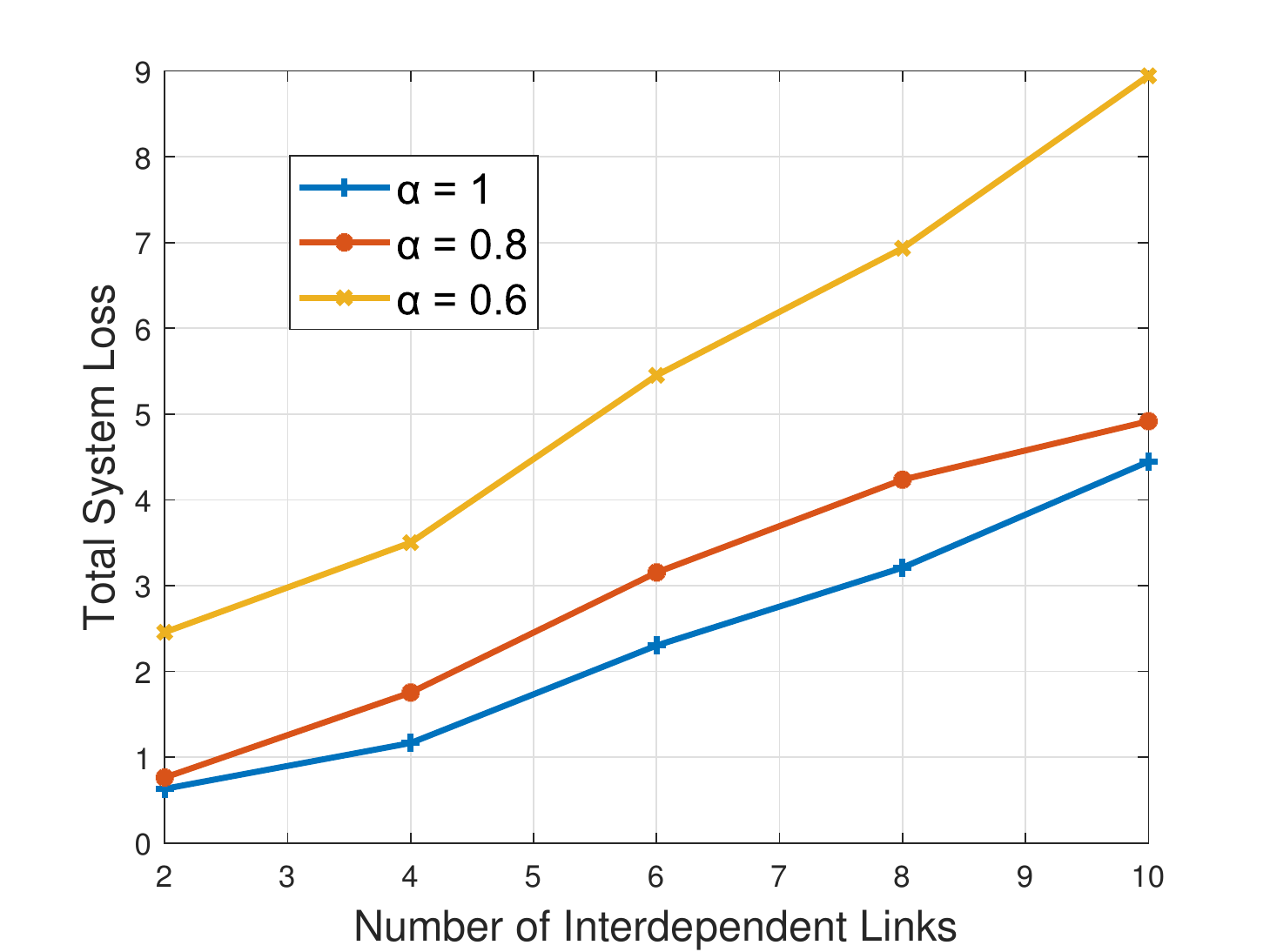}
\caption{The effect of increasing the degree of interdependency on the total system loss. Such effect is more pronounced when the defender is more behavioral.}
  \label{fig:Interdependency_Effect_Scada}
\end{minipage}
  \vspace{-2mm}
\end{figure}

\begin{figure*}[t] 
\centering
\begin{subfigure}[t]{.32\textwidth}
\centering
   \includegraphics[width=\linewidth]{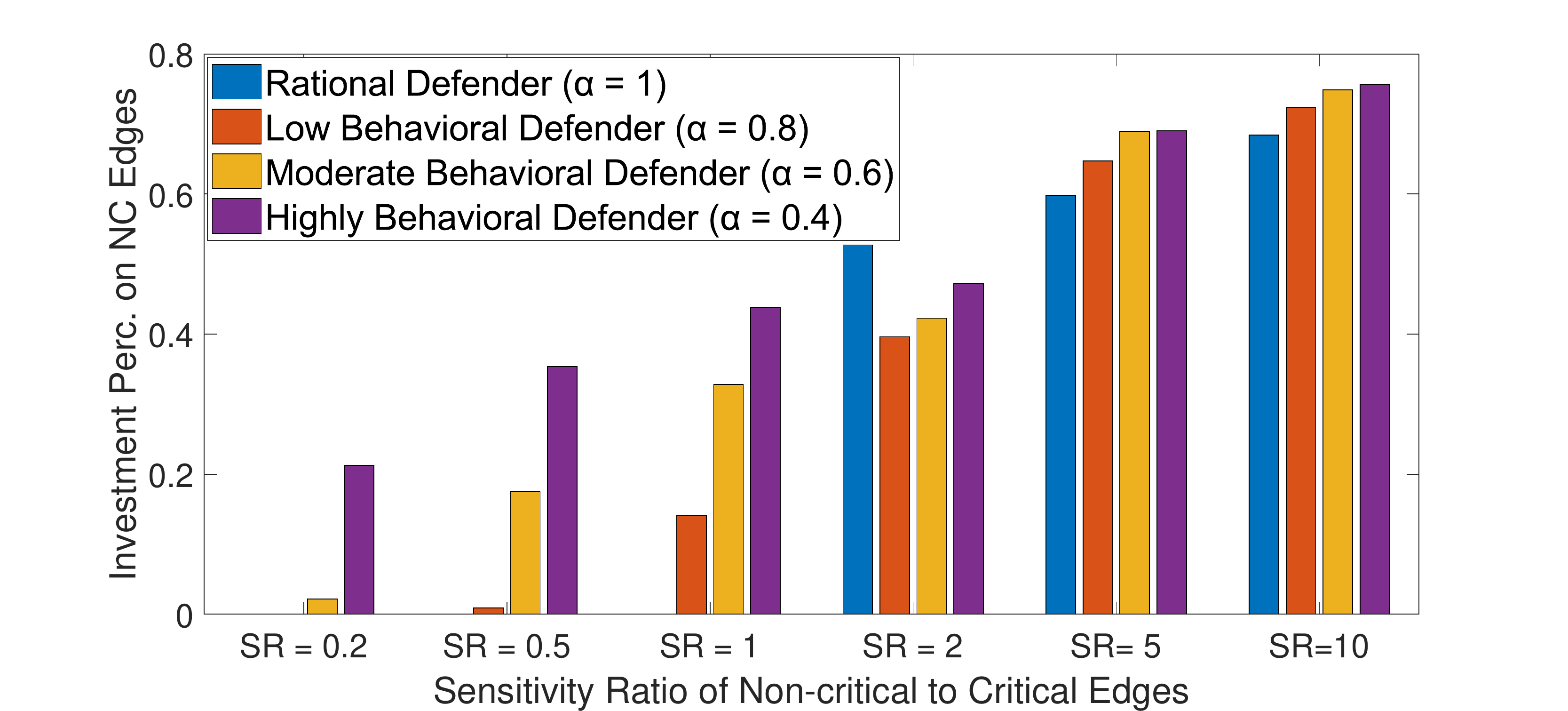}
  \caption{Percentage of investments on non-critical edges for different edge sensitivities ratios (SR).}
  \label{fig:sensitivity-exp-Scada}
 \end{subfigure}\hfill
 \begin{subfigure}[t]{.32\textwidth}
\centering
  \includegraphics[width=\linewidth]{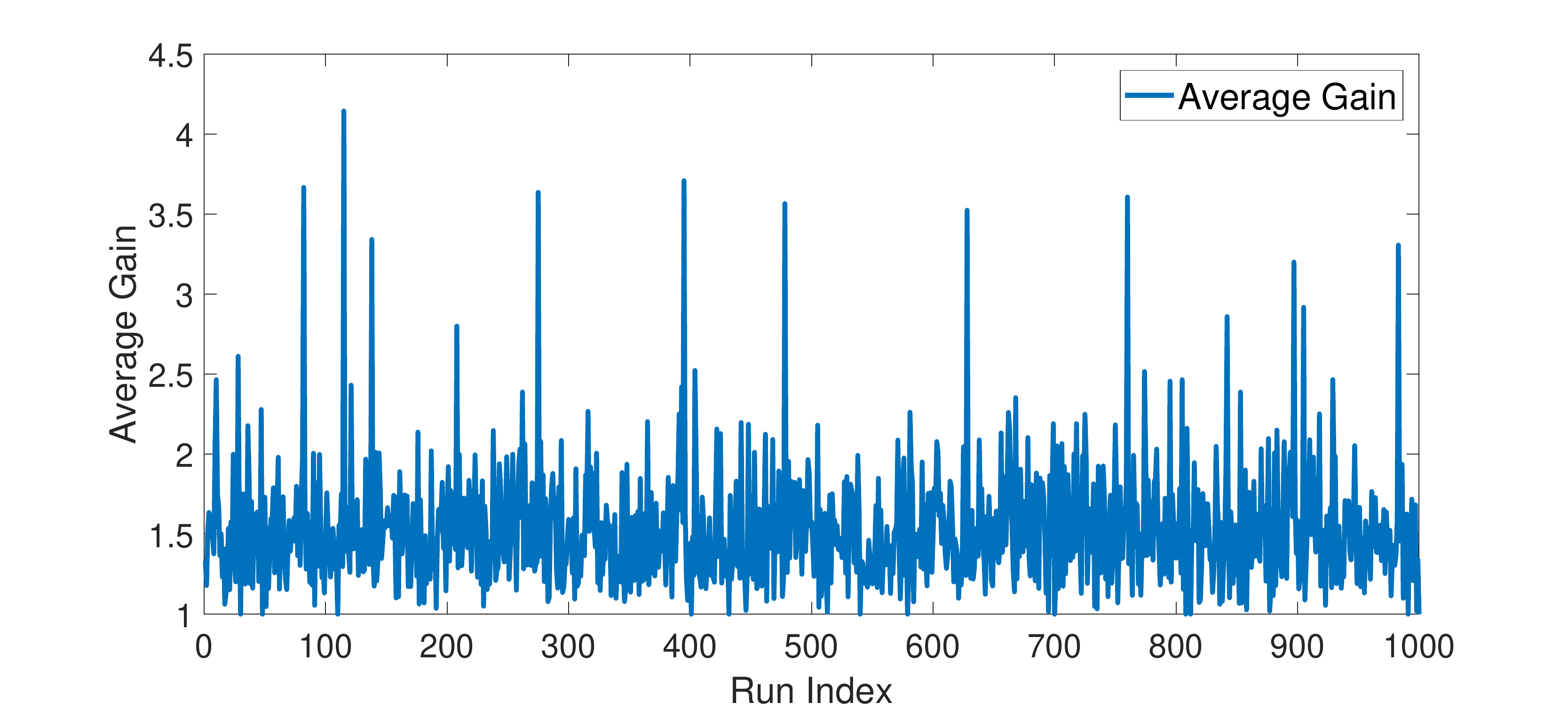}
  \caption{Average gain for different random attack success probabilities over all of the edges.}
  \label{fig:effect_all_probabilites_scad}
 \end{subfigure}\hfill
 \begin{subfigure}[t]{.32\textwidth}
\centering
 \includegraphics[width=\linewidth]{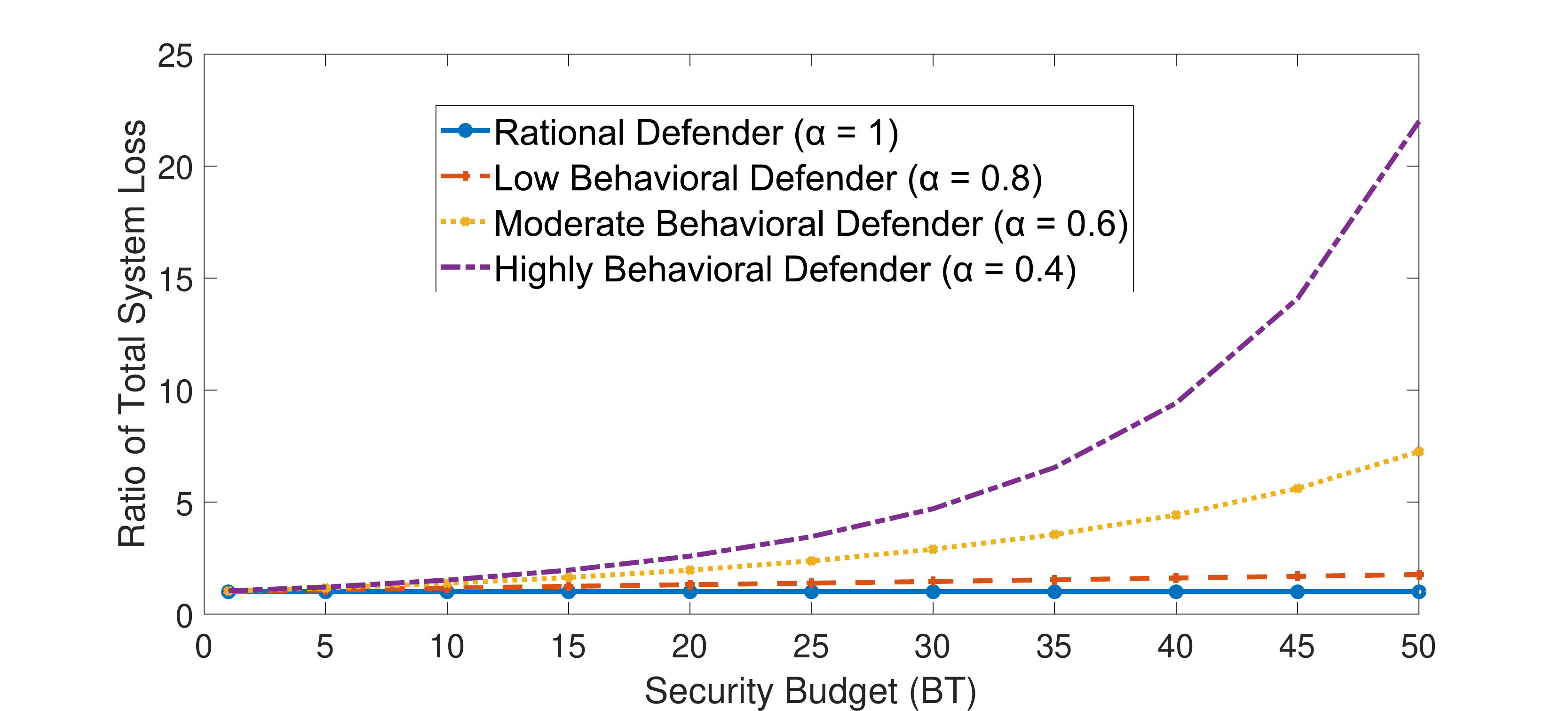}
  \caption{Effect of behavioral investments for different security budgets.}
  \label{fig:effect_all_security_budgets_scada}
 \end{subfigure}
 \caption{(a) The effect of edges' sensitivities on investments for different behavioral levels. (b) The average gain of rational decision-making for randomly chosen baseline probabilities of successful attacks. (c) The effect of sub-optimal investments for different choices of security budget. }
\label{fig:parameters_effects}
\end{figure*}

\textbf{1) Effect of defense mechanism:}
We observe the merits of cooperation (i.e., joint defense) in decreasing the total loss to the defenders as shown in Figure~\ref{fig:Defense_mechanisms}. The effect is more pronounced for a higher degree of behavioral bias of the defenders. For example, at moderate budget ($BT = 20$), the relative decrease in total system loss due to joint defense at $\alpha_1 = \alpha_2 = 0.4$ is 25\% while $\alpha_1 = \alpha_2 = 0.8$, the decrease is lower (10\%). Thus, as the defenders exhibit higher degree of cognitive bias, it is more advantageous to adopt joint defense mechanisms. 

\textbf {2) Interdependency among different defenders.} 
Here, we observe effect of interdependency between defenders on the security of the SCADA system. In the SCADA system, the degree of interdependency increases if assets from one subnetwork can access assets in the other, without going through the Corporate or Vendor nodes. For example, if the attacker gets access to Control unit $1$, this enables her to compromise RTU2 as well, in addition to RTU1. Figure \ref{fig:Interdependency_Effect_Scada} illustrates this effect---as the number of interdependent edges between the two defenders increases, the total system loss increases in both non-behavioral and behavioral security games. The highest level of interdependency is when there are two edges between DMZ1 and DMZ2, between Control1 and Control2, and the controller to the 3 RTUs of the other defender. An example of this phenomenon is that if both defenders are non-behavioral and the level of interdependency is the highest, the total system loss is higher by $ 462\% $ over the case of the lowest level of interdependency ($2$ interdependent links). We also see that as the interdependency between the different defenders increases, the suboptimal security decisions have greater adverse impact on the total system loss.

\textbf{3) Sensitivity of edges to investments:} 
We next consider the effects of different sensitivities of edges to security investments. Recall that higher sensitivity edges are those for which the probability of successful attack decreases faster with each unit of security investment. We show the result in Figure~\ref{fig:sensitivity-exp-Scada} by using as the independent variable the ratio of sensitivity of non-critical to critical edges. First, assume critical edges correspond to mature systems that are already highly secure and difficult to secure further. For our model, this translates to high (resp. low) $s_{i,j}$ for non-critical (resp. critical) edges. We observe that as the sensitivity ratio increases, all defenders put more investments on the non-critical edges, but the increase is slower in behavioral defenders. However, lower sensitivity ratio will result in investing almost all budget on these critical edges, even for behavioral defenders.

\textbf{4) Baseline probabilities of successful attacks:} 
We show that the gain of rational vs. behavioral investments exists for any combination of baseline probabilities by performing 1000 runs and in each run, for each edge, we draw the baseline probability of successful attack on that edge from a uniform distribution  $p^{0}_{i,j} \sim \mathcal{U}(0,1]$. We consider a symmetric budget distribution and medium security budgets. Figure~\ref{fig:effect_all_probabilites_scad} shows that the gain for rational over behavioral decision-making (with mean 1.53X) exists for any randomly chosen baseline probability of successful attacks. 

\textbf{5) Amount of security budget:} 
We next show that the total system loss of rational defenders is less than that of behavioral defenders for any choice of security budget (as shown in Figure~\ref{fig:effect_all_security_budgets_scada}).

\textbf{6) Security budget distribution among defenders:}
Finally, we analyze the effect of asymmetric budget distribution between the defenders facing the attacker. Figure~\ref{fig:asymm_scada} (in Appendix~\ref{app: evaluation_extended}) illustrates the total loss as a function of the fraction of defender 1's budget. For the individual-defense loss, we observe that the suboptimality of behavioral decision-making is more pronounced with higher budget asymmetries. For example, if defender 1 has 20\% of the total budget, the relative increase in total loss from $\alpha = 1$ to $\alpha = 0.4$ is 25\%. In contrast, the same change of $\alpha$ when the budget is symmetric results in only a 6\% relative increase in the total loss. This observation can be explained by two facts. First, with suboptimal behavioral allocation, the poorer defender wastes even her constrained budget on non-critical edges. Second, the richer player also allocates her resources suboptimally. This leads to this magnified relative increase in losses under budget asymmetry.

In the interest of space, we present system parameters evaluation of DER.1, which has similar insights as SCADA, in Appendix~\ref{app: evaluation_extended}.

\section{Limitations and Discussion} \label{sec:Discussion}

\textbf{Guiding security decision-makers:} We believe that our work opens up a new dimension of {\em intervention} in securing interdependent systems. Our framework allows a quantification of the improvements in security that can be obtained by training security professionals to reduce their behavioral biases. In this context, we can quantitatively show the decision-maker the improvement in system security when moving from her current (sub-optimal) investments to that given by a (rational) algorithm (e.g., \name with $\alpha = 1$). Furthermore, our framework can guide security audits by system operators of large-scale interdependent system, by allowing the operator to investigate subsystems within the system where sub-optimal security investments might have been made by subordinates operating those subsystems. While such an operator may not be able to check every single aspect of every subsystem, she may be able to “zoom in” to portions of the overall system where an audit may be warranted due to evidence of sub-optimality from our framework. 

\noindent \textbf{Behavioral level of the attacker:} We assume that defenders perceive the attacker as non-behavioral; in reality the attacker can be behavioral as well. Our assumption of a  non-behavioral attacker gives the worst case loss  for the system; as a behavioral attacker may not choose the path of true highest vulnerability due to probability misperceptions. This can open the interesting question ``how a rational defender, who uses the security investments recommended by \name, can deceive a behavioral attacker to choose harder attack paths?'' This can help the defender to misguide the attacker and make the target system more secure.

\noindent \textbf{Multi-hop dependence:}   In several cybersecurity scenarios, the ease of an attacker in achieving an attack goal depends not just on the immediate prior attack step but on steps farther back. In such scenarios, the simpler formulation of using probabilities on each edge and assuming independence of the events of traversing the different edges can lead to inaccurate estimates. However, we follow several prior works (e.g.,~\cite{modelo2008determining, xie2010using}) that leveraged the property that in most cases, a node has the highest dependence on the previous node, in order to build computationally tractable analysis tools. Moreover, to handle this issue in our model, the notion of {\em k-hop dependence} \cite{maheshwari2007detecting} can be used, whereby the probability of reaching a particular node depends  nodes up to $k$ hops away. 

\vspace{-2mm}
\section{Related Work}\label{sec:lit-review}

\noindent {\bf Security in interdependent systems}: 
The problem of securing systems with interdependent assets has been handled in several prior works~\cite{modelo2008determining, xie2010using}. The common theme is that a successful attack to one asset may be used to compromise a dependent asset. The notion of attack graphs \cite{homer2013aggregating} is a popular abstraction for capturing the security interdependencies. The specific works differ in what the assets are (physical or virtual, resource-constrained nodes, networking assets, etc.), the level of observability into the states of the assets, and the probabilistic reasoning engine used. Our work here differs from these works 
in that the prior work creates algorithms to make the security control decisions, while we are considering humans with cognitive biases making these decisions.

\noindent{\bf Game-theoretic modeling of security}: 
 Game theory has been used to describe the interactions between attackers and defenders and their effects on system security. 
A commonly used model in this context is that of two-player games, where a single attacker attempts to compromise a system controlled by a single defender \cite{roy2010survey,alpcan2006intrusion}. Game theoretic models have been further used in  \cite{yan2012towards} to study the interaction between one defender and (multiple) attackers attempting Distributed Denial of Service  attacks. 
Game theoretic models have also been proposed for studying critical infrastructure security (See the survey~\cite{laszka2015survey}).
The major difference of our work with all aforementioned literature is that existing work has focused on classical game-theoretic models of rational decision-making, while we analyze behavioral models of decision-making. 

\noindent \textbf{Human behavior in security and privacy:}
Notable departure from classical economic models within the security and privacy literature is \cite{acquisti2009nudging}, which identifies the effects of behavioral decision-making on individual's personal privacy choices. The importance of considering similar models in the study of system security has been recognized in the literature \cite{cranor2008framework}. Prior works~\cite{redmiles2018dancing,anderson2012security} considered models from behavioral economics in the context of security applications. However, these works are based only on psychological studies \cite{anderson2012security}  and human subject experiments \cite{redmiles2018dancing} for end-user. Our work differs from these in that we explore a rigorous mathematical model of defenders' (decision-makers) behavior, model the interaction between multiple defenders (in contrast to the study of only one defender for all of these studies), and consider interdependent assets (in contrast to these studies which reason about binary decisions on isolated assets). To the best of our knowledge, the exceptions that provide a theoretical treatment of behavioral decision-making in certain specific classes of interdependent security games are \cite{7544460,9030279,sanjab2017prospect,abdallah2020behavioral}. These works, however, are theoretical in scope and do not consider the more realistic attack scenarios and types that we consider, do not validate bias of decision-makers via subject experiments, and don't consider multi-round setups or learning algorithms that we consider here. 

\noindent {\bf Multi-round in Security}:
Reinforcement-based learning models have been used in literature where players' strategies receive reinforcement related to the payoffs they earn and adjust their moves over time seeking higher payoffs.
Specifically, \cite{feltovich2000reinforcement} proposed reinforcement learning for an environment with only two possible actions. Such Reinforcement-based learning models have been used in different security applications such as the robustness of smart grid~\cite{ni2019multistage}. Our work differs from these works that we guide the behavioral decision-maker towards rational decision-making where the reinforcements are received from the true loss that the defender accrues when investing with behavioral bias. Likelihood method of discovering attack paths using Bayesian attack graphs has been proposed in \cite{xie2010using}. However, to the best of our knowledge, no previous work has the idea of minimizing the adapted defender's cost and generate optimal allocations each round while weighting attack paths based on previous rounds that we consider in our Hybrid-learning algorithm.

\section{Conclusion}\label{sec:conclusion}

We presented \emph{behavioral security games} to study the effects of human behavioral decision-making on the security of interdependent systems with multiple defenders where we model stepping-stone attacks by the notion of {\it attack graphs}. While behavioral decision-makers tend to allocate their budget across the network, \name helps decision-makers concentrate their budget on critical edges to make the system more secure.  We performed a controlled subject experiment to validate our behavioral model. In multi-round setups, we proposed different learning algorithms to guide behavioral decision-makers towards optimal decisions. We evaluated \name on five real case studies of interdependent systems where we studied the effects of several system parameters. The insights gained from our analysis would be useful for configuring real-world systems with optimal parameter choices and guiding behavioral decision-makers toward rational decision-making that can ultimately lead to improvements in interdependent systems' security.

\bibliographystyle{ACM-Reference-Format}
\bibliography{ccs-sample.bib}

\appendix
\appendix

\section{Convexity of Total Loss Function}\label{sec:proof-convexity}
\begin{lemma}
Let the probability successful function $p_{i,j}(x_{i,j})$ be twice-differentiable and log-convex. Then, the total loss function in \eqref{eq:defender_utility_edge} is convex.
\end{lemma}
\begin{proof}
 We drop the subscript $i,j$ in the first part of this analysis for better readability. Now, beginning with the probability weighting function defined in \eqref{eq:prelec}, we have $ w(p(x)) = (g \circ h)(x)$, where $ g(x) = \exp(-x) $ and  $ h(x) = (-\log(p(x)))^\alpha$. 
Now, we prove that $h(x)$ is concave.
\vspace{-2mm}
\begin{align*}
h''(x) &= \alpha (\alpha-1) (-\log(p(x)))^{\alpha-2} \frac{(p'(x))^2}{(p(x))^2}\\
& + \alpha (-\log(p(x)))^{\alpha-1} \left[ \frac{(p'(x))^2 - (p(x))(p''(x))}{(p(x))^2} \right].
\end{align*}
Since $0 \leq p(x)\leq 1$, we have $0 \leq -\log(p(x)) \leq \infty $ for all $x$. Moreover, $0 < \alpha \leq 1$ and thus the first term in the R.H.S. of $h''(x)$ is negative. Also, since $p(x)$ is twice-differentiable and log-convex, $(p'(x))^2 < (p(x))(p''(x))$, which ensures that the second term is also negative. Therefore, $h(x)$ is concave.
Since $ g(x) $ is convex and non-increasing while $ h(x) $ is concave, $ w(p(x)) $ is convex. 

Now, since $w(p_{i,j}(x_{i,j}))$ is monotone and convex, thus $ \prod_{(v_{i},v_{j}) \in P} w(p_{i,j}(x_{i,j})) $ is convex. Moreover, the maximum of  a set of convex functions is also convex. Finally, since the total loss function $ C_{k}(x_k,\mathbf{x}_{-k}) $ is a linear combination of convex functions, the total loss function defined in \eqref{eq:defender_utility_edge} is convex.
\end{proof}

\begin{table}[t] 
\caption {A summary of the \name's variables.}%
\resizebox{\columnwidth}{!}
{
 \begin{tabular}{|c | c |} 
 \hline
 Variable & Description  \\ 
 \hline
 $G = (\mathcal{V}, \mathcal{E})$ &  Attack graph of the system with set of nodes $\mathcal{V}$ and the set of edges $\mathcal{E}$\\
  \hline
 $p_{i,j}^{0}$ & Baseline probability of successfully compromising asset $v_{j}$ starting at $v_{i}$\\
\hline
$v_{s}$ & The attacker's source node \\
\hline
$V_{m}$ & The set of all critical assets in the system \\
\hline
$V_{k}$ & The set of critical assets under control of the defender $D_k$ \\
\hline
 $P$ & Directed path from $v_s$ to $v_m\in V_m$\\
\hline
 $P_{m}$ & The set of all directed paths from $v_s$ to
$v_m$ \\
\hline
$D$ & The set of all defenders of the network\\
\hline
 $D_{k}$ & A defender who controls a set of nodes $V_{k}$\\
 \hline
 $\mathcal{E}_k$ & The edge set defended by $D_k$ \\
 \hline
$L_{m}$ & Financial loss of defender $D_{k}$ if the asset $v_{m} \in V_{k}$ is compromised \\
 \hline
 $B_{k}$ & The security budget of defender $D_{k}$\\
 \hline
  $x^{k}_{i,j}$ & The security investment of defender $D_{k}$ on the edge $(v_i,v_j)$ \\
 \hline
 $x_{i,j}$ & Total investments on the edge $(v_i,v_j)$ by all eligible defenders\\
 \hline
  $x_k$ & The vector of investments by defender $D_k$  \\
 \hline
 $\mathbf{x}_{-k}$ & The vector of investments by defenders other than $ D_{k}$ \\      
 \hline
  $s_{i,j}$ & Sensitivity of edge $(v_i,v_j)$ to investments \\
\hline
 $C_{k} (x_k,\mathbf{x}_{-k})$ & Cost (Total loss) function of defender $D_{k}$\\
 \hline
 $p_{i,j} (x_{i,j})$ & True probability of successful attack on edge $(v_i,v_j)$ \\
 & given security investments $ x_{i,j} $ \\
\hline
 $w(p_{i,j} (x_{i,j}))$ & Perceived probability of attack on the edge 
 $(v_i,v_j)$ \\
\hline
\end{tabular}
}
\label{tb:system_variables}
\vspace{-1mm}
\end{table}

\section{Motivational Example with different sensitivities}\label{app:sensitivity}
In the above example, we assumed all edges have the same sensitivity to investments. In cases where critical edges have equal or higher sensitivity than non-critical edges, the same insight as above holds.
Specifically, when edge $(v_i, v_j)$ has sensitivity $s_{i,j}$, one can verify (using KKT conditions) that the optimal investments by a behavioral defender are given by 
\begin{equation*}
\begin{aligned}
x_{1,2} &= x_{2,4} = x_{1,3} = x_{3,4} = 2^{\frac{1}{\alpha-1}} \left(\tfrac{s_{i,j}}{s_{s,1}}\right)^{\frac{\alpha}{1-\alpha}} x_{s,1} .\\
x_{s,1} &= \left(\tfrac{s_{s,1}}{s_{4,5}}\right)^{\frac{\alpha}{1-\alpha}}   x_{4,5};\hspace{1mm} x_{4,5} = B - \sum_{\forall (i,j)\neq (v_4,v_5) } x_{i,j}.
\end{aligned}
\end{equation*}
The insight here is that the investment decision has two dimensions: behavioral level and sensitivity ratio of non-critical edges to critical edges. Specifically, as the defender becomes more behavioral, she puts less investments on edges with higher sensitivity.

\section{Human Subject Demographics}\label{app: human-exp-extended}
The 145 human subjects in our experiment are comprised of 78 males (53.79\%) and 67 females (46.21\%). They belong to various majors on campus, with the three largest being Management/Business (24.8\%), Engineering (24.2\%), and Science (23.5\%). Regarding year in college, 6.9\% are 1st year, 13.1\% are 2nd year, 21.38\% are 3rd year, 35.86\% are 4th year, and 22.76\% are graduate students. Regarding the GPA distribution, 44.83\% have GPAs between 3.5 and 4, 35.17\% between 3 and 3.5, and 17.93\% between 2.5 and 3.

\section{Convergence of Reinforcement learning of behavioral level}\label{app:RL_convergence_analysis}

\begin{lemma}\label{lemma: RL_convergence}
Let $N_1$ and $N_{\alpha_i}$ represent the number of rounds in which the defender chose to invest rationally and with behavioral level $\alpha_i$, respectively. Let $\hat{C}_{opt}$ and $\hat{C}_{i}$ be the total real loss incurred by the defender when investing rationally and with behavioral level $\alpha_i$, respectively. Then, we have
\begin{enumerate}
\item If $q^{0}(\alpha_{i}) = q^{0}(1)$, then Algorithm~\ref{alg:rl} converges to rational behavior in the round $N_R$ if $\frac{\hat{C}_{max} - \hat{C}_{opt}}{\hat{C}_{max} - 
    \hat{C}_{i}} >> \frac{N_{\alpha_i}}{N_1}$.

\item If  $q^{0}(\alpha_{i}) = A_i$ and $q^{0}(1) = B$, then Algorithm~\ref{alg:rl} converges to rational behavior in the round $N_R$ if  $(N_{1} - N_{\alpha_i}) \hat{C}_{max} + N_{\alpha_i} (\hat{C}_{i})  - N_{1} (\hat{C}_{opt})>> A_i - B$. 
\end{enumerate}
\end{lemma}
\begin{proof}
We first calculate the propensities for each behavioral level. From Algorithm~\ref{alg:rl}, we have
\begin{align*}
   q^{N_R}(1)  &= q^{0}(1) + N_1 (\hat{C}_{max} - \hat{C}_{opt}) \\
   q^{N_R}(\alpha_{i})  &= q^{0}(\alpha_{i}) + N_{\alpha_i} (\hat{C}_{max} - \hat{C}_{i})
\end{align*}

\noindent (i) To reach convergence, $\forall \alpha_{i} \neq 1$ and $q^{0}(\alpha_{i}) = q^{0}(1)$, we have
\begin{align*}
q^{N_R}(1) >> q^{N_R}(\alpha_{i}) 
\iff & N_1 (\hat{C}_{max} - \hat{C}_{opt}) >> N_{\alpha_i} (\hat{C}_{max} - \hat{C}_{i})  \\
\iff & \frac{\hat{C}_{max} - \hat{C}_{opt}}{\hat{C}_{max} - \hat{C}_{i}} >> \frac{N_{\alpha_i}}{N_1}
\end{align*}
(ii) With a similar argument to (i),  $\forall \alpha_{i}  \neq  1 $ and $q^{0}(\alpha_{i}) \neq q^{0}(1)$ where $q^{0}(\alpha_{i}) = A_i$ and $q^{0}(1) = B$, we have 
\begin{small}
\begin{align*}
q^{N_R}(1) >> q^{N_R}(\alpha_{i}) \\
\iff & (N_{1} - N_{i}) \hat{C}_{max} + N_i (\hat{C}_{i})  - N_{1} (\hat{C}_{opt})>> A_i - B
\end{align*}
\end{small}
Note that in all of the possible cases, $\hat{C}_{i} - \hat{C}_{opt} >> 0$ ensures convergence under any choice of $A$ and $B \neq 0$ which is realistic where the real loss associated with suboptimal investments decisions is much higher compared to the real loss associated with optimal (i.e., rational) investments decisions. 
\end{proof}

\begin{table}
\caption {Baseline probability of successful attack for vulnerabilities in SCADA and DER.1 failure scenarios.}
\label{tbl:cvss_cve_der_scada}
\centering
\resizebox{\columnwidth}{!}
{%
\begin{tabular}{|l|l|l|l|}
\hline
\multicolumn{1}{|l|}{\text{\bf Vulnerability (CVE-ID) }}
& \multicolumn{1}{l|}{\bf Edge(s)}
& \multicolumn{1}{l|}{\bf Attack Vector}
& \multicolumn{1}{l|}{\bf Score}\\
\cline{1-4}
\hline
\multicolumn{4}{|l|}{\bf SCADA application} \\ 
\hline
Control Unit (CVE-2018-5313) & (Vendor,Control1),(Vendor,Control2) & Local & 0.78 \\
\hline
Remote authentication (CVE-2010-4732) & (S, Vendor) & Network  & 0.9 \\  
\hline
Remote cmd injection (CVE-2011-1566) &  (Control,RTU1),(Control,RTU2) & Network  &  1.0\\
\hline
Authentication bypassing (CVE-2019-6519) & (Corp,DMZ1),(Corp,DMZ2) & Network & 0.75\\
\hline
\multicolumn{4}{|l|}{\bf DER.1 application} \\ 
\hline
Physical access (CVE-2017-10125) & ($w_{9},w_{7}$),($w_{18},w_{16}$) & Physical & 0.71 \\  
\hline
Network access (CVE-2019-2413) & ($w_{9},w_{8}$),($w_{18},w_{17}$) & Network & 0.61 \\ 
\hline
Software access (CVE-2018-2791) & ($w_{7},w_{6}$),($w_{8},w_{6}$) & Network & 0.82 \\ 
\hline
Sending cmd (CVE-2018-1000093) & ($w_{6},w_{5}$),($w_{15},w_{14}$) & Network &  0.88 \\ 
\hline
\end{tabular}%
}%
\end{table}

\section{Attack Scenarios of Case Studies}\label{app:systems_overview}
In this section, we provide explanations of the system overview and equivalent attack graph for the DER.1, E-commerce, VOIP, and IEEE 300-BUS, respectively.
\begin{figure*}
\begin{minipage}[t]{\linewidth}
\begin{minipage}{0.45\linewidth}
  \centering
  \includegraphics[width=\linewidth]{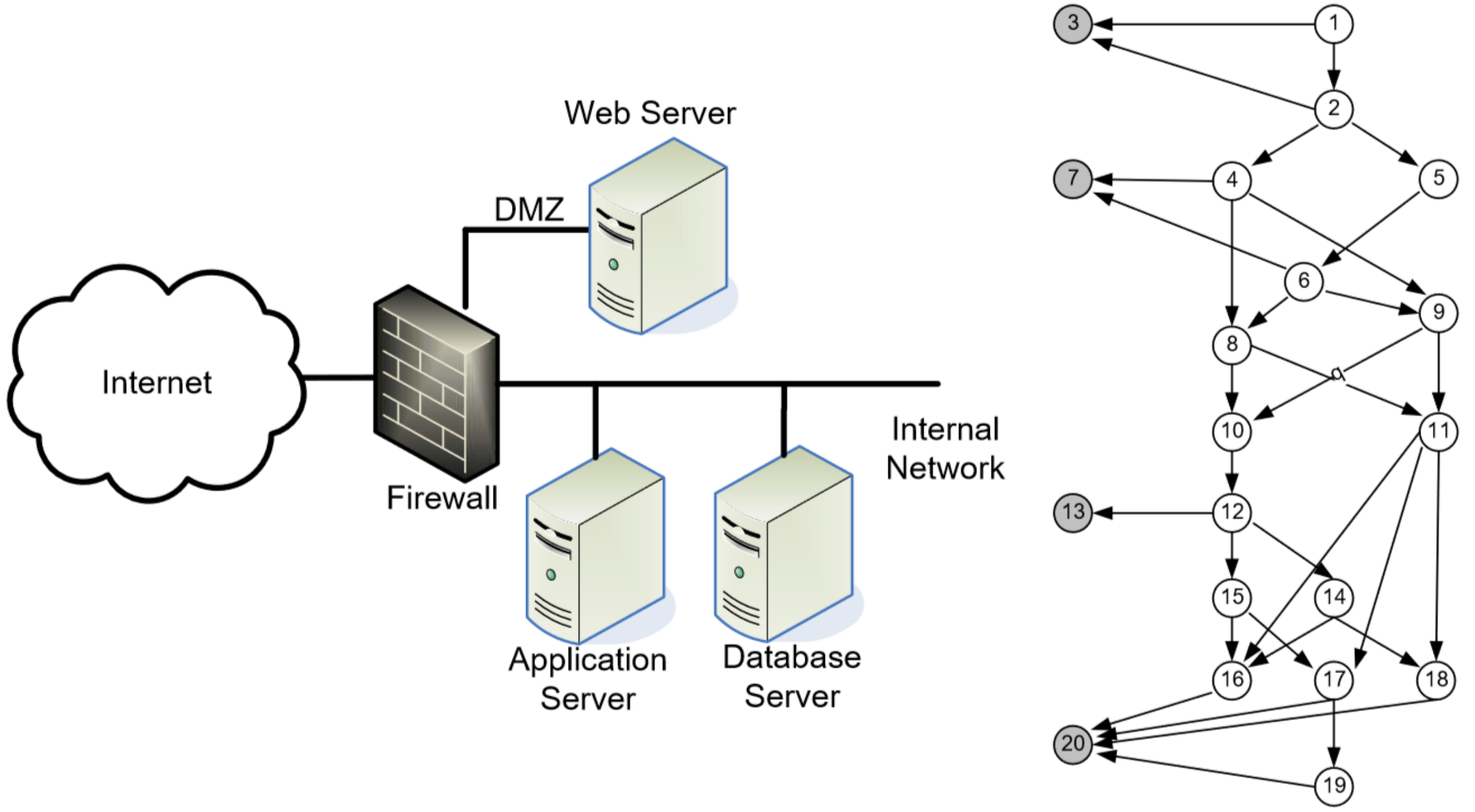}
  \caption{A high level network overview of the E-commerce (on the left) adapted from \cite{modelo2008determining}. The resultant attack graph (on the right).}
  \label{fig:Ecommerce_High_level_overview}
\end{minipage} \hfill
\begin{minipage}{0.45\linewidth}
\centering
  \includegraphics[width=0.9\linewidth]{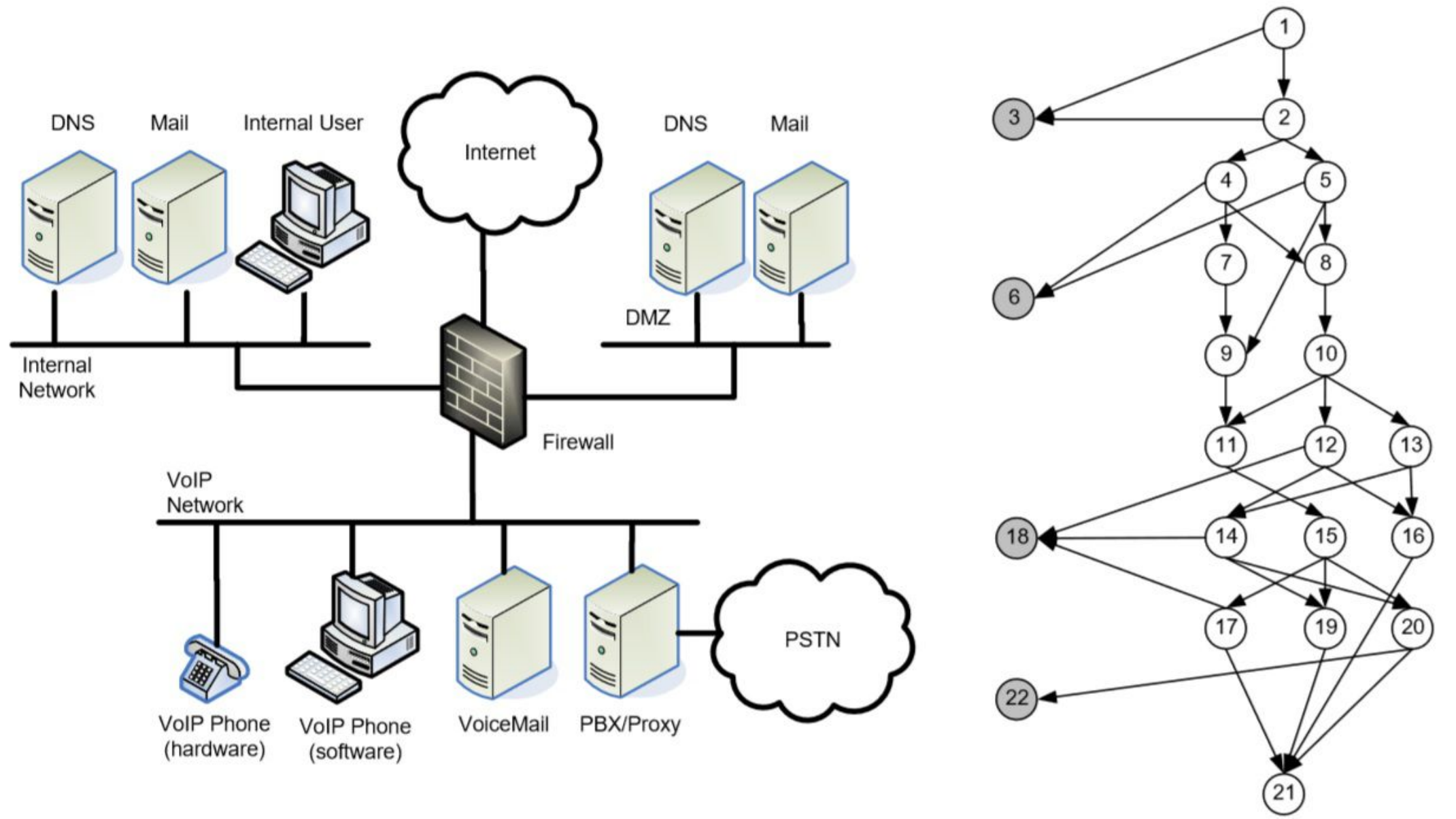}
  \caption{A high level network overview of the VoIP (on the left) adapted from \cite{modelo2008determining} and its resultant attack graph  (on the right).}
  \label{fig:voip_High_level_overview}
\end{minipage} 
\end{minipage}
\end{figure*}

\begin{figure*}
\begin{minipage}[t]{\linewidth}
\begin{minipage}{0.47\linewidth}
\centering
 \includegraphics[width=\linewidth]{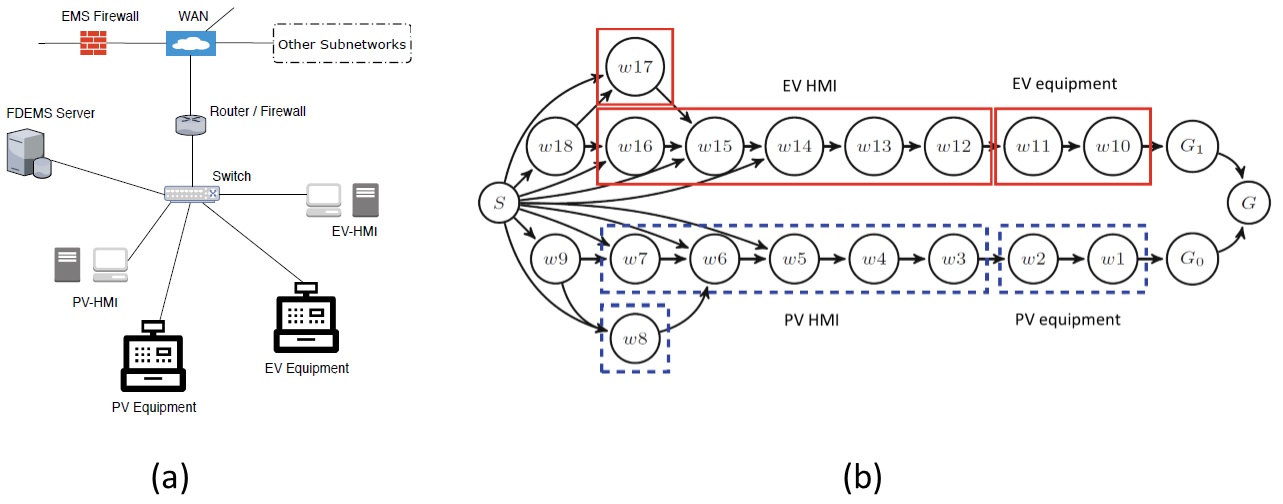}
  \caption{ A network level overview in (a) of a Distributed Energy Resource (DER) system (adapted from \cite{jauhar2015model}). In (b), the equivalent attack graph of that failure scenario is shown.}
  \label{fig:System_overview_DER1}
\end{minipage}\hfill
\begin{minipage}{0.47\linewidth}
\centering
 \includegraphics[width=0.8\linewidth]{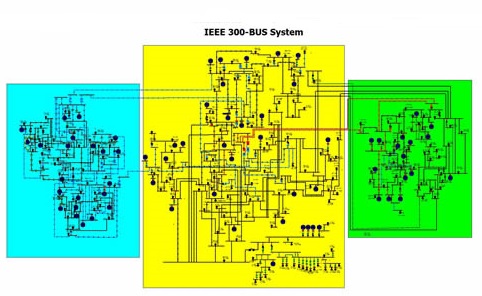}
  \caption{ A high level overview of the IEEE 300- BUS (adapted from \cite{khanabadi2012optimal}). Each area has a different color.}
  \label{fig:System_overview_300BUS}
\end{minipage}
\end{minipage}
\end{figure*}

\subsection{DER.1}
 
\textbf{System Description:} The US National Electric Sector Cybersecurity Organization Resource (NESCOR) Technical Working Group has proposed a framework for evaluating the risks of cyber attacks on the electric grid. A distributed energy resource (DER) is described as a cyber-physical system consisting of entities such as generators, storage devices, and electric vehicles, that are part of the energy distribution system. The DER.1 failure scenario has been identified as the riskiest failure scenario affecting distributed energy resources according to the NESCOR ranking. As shown in Figure~\ref{fig:System_overview_DER1}, there are two critical equipment assets: a PhotoVoltaic (PV) generator and an electric vehicle (EV) charging station. Each piece of equipment is accompanied by a Human Machine Interface (HMI), the only gateway through which the equipment can be controlled. The DER.1 failure scenario is triggered when the attacker gets access to the HMI. The vulnerability of the system may arise due to various reasons, such as hacking of the HMI, or an insider attack. Once the attacker gets access to the system, she changes the DER settings and gets physical access to the DER equipment so that they continue to provide power even during a power system fault. Through this manipulation, the attacker can cause physical damage to the system. 

\subsection{E-commerce}
\textbf{System Description:} In that E-commerce system (shown in Figure~\ref{fig:Ecommerce_High_level_overview}), all servers are running a Unix-based operating system. The web server sits in a demilitarized zone (DMZ) separated by a firewall from the other two servers, which are connected to a network not accessible from the Internet. All connections from the Internet and through servers are controlled by the firewall. Rules state that the web and application servers can communicate, as well as the web server can be reached from the Internet. Here, the attacker is assumed to be an external one and thus her starting point is the Internet which uses stepping-stone attacks with the goal of having access to the MySQL database (specifically access customer confidential data such as credit card information), represented by node $19$ in the attack graph. For this system, we follow the attack graph generated by \cite{modelo2008determining} (shown in Figure~\ref{fig:Ecommerce_High_level_overview}), which is based on the vulnerabilities associated with specific versions of the particular software, and are taken from popular databases. 

\subsection{VoIP}

\textbf{System Description:} As shown in Figure~\ref{fig:voip_High_level_overview}, the VoIP system is composed of three zones; a DMZ for the servers accessible from the internet cloud, an internal network for local resources (e.g., computers, mail server and DNS server), and an internal network that is consisted of only VoIP components. This architecture follows the security NIST guidelines for deploying a secure VoIP system. In this context, the VoIP network consists of a Proxy, voicemail server and software-based and hardware-based phones. The firewall has the rules to control the traffic between the three zones. Note that the DNS and mail servers in the DMZ are the only accessible hosts from the Internet. The PBX server can route calls to the Internet or to a public-switched telephone network (PSTN). The ultimate goal of this multi-stage attack is to eavesdrop on VoIP communication.

\subsection{IEEE 300 BUS}
\textbf{System Description:} Finally, we consider the widely used benchmark IEEE 300 bus power grid network \cite{khanabadi2012optimal}. We define the network itself as the interdependency graph where each node represents a bus (i.e., the network has 300 nodes), and the physical interconnection between the buses represent the edges. Each bus has generators and/or load centers associated with it. As shown in Figure~\ref{fig:System_overview_300BUS}, the 300 bus network data divides the buses or nodes into 3 different regions containing 159, 78 and 63 nodes respectively. We assume that each region is managed by an independent entity or defender. The defenders want to protect the buses within their region that contain the generators; each generator
is valued at its maximum generation capacity. The attacker can directly access three nodes (specifically, bus 39, 245 and 272).

\begin{figure}[t] 
\centering
  \includegraphics[width=0.9\linewidth]{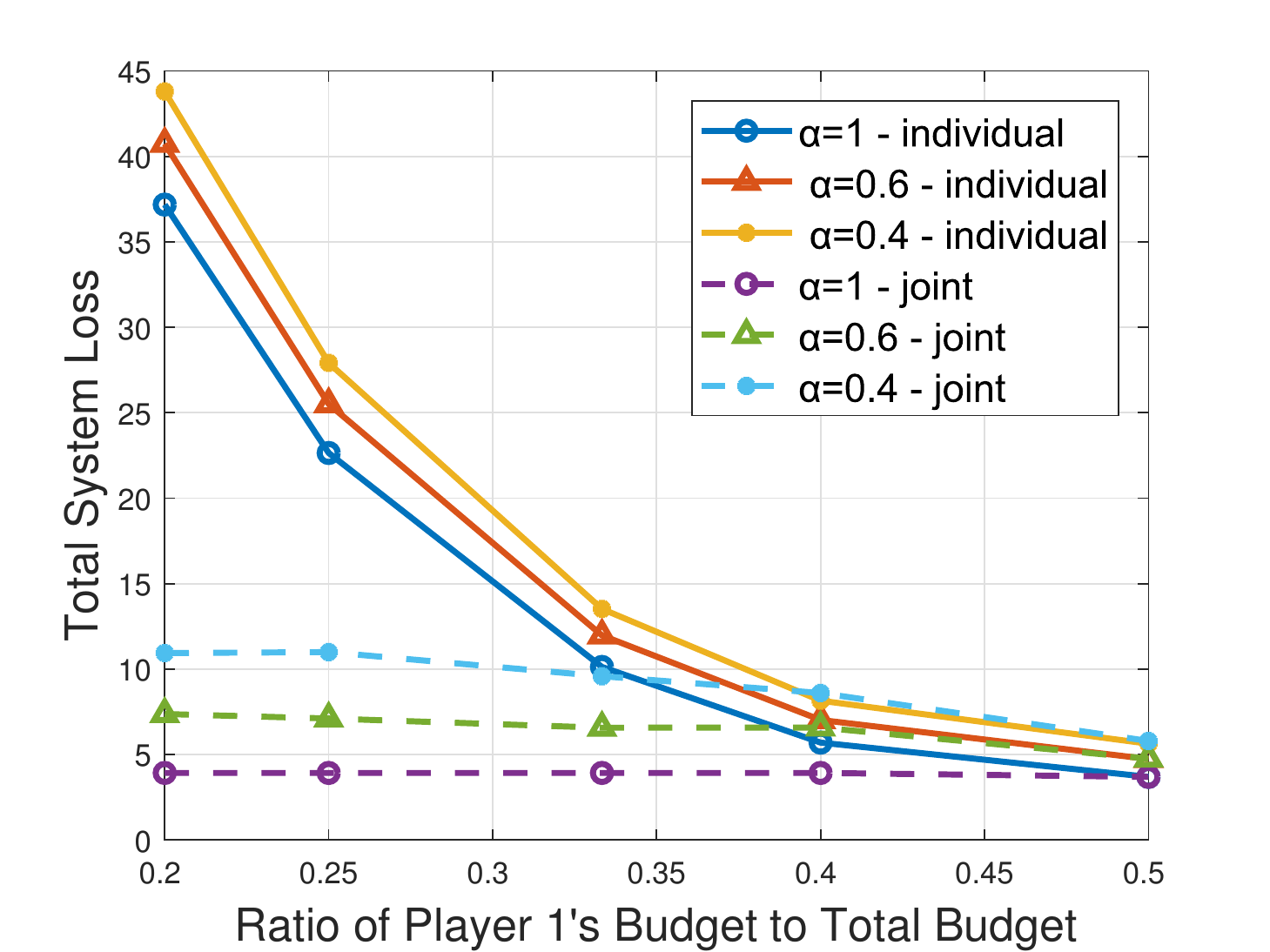}
 \caption{The total system loss as a function of the fraction of defender 1's budget. We observe that joint defense outperforms individual defense at higher budget asymmetry.}
  \label{fig:asymm_scada}
\end{figure}

\section{Evaluation-Extended}\label{app: evaluation_extended}

\subsection{Muti-Defenders: DER.1}
We present the system parameters results (shown in Figure~\ref{fig:extended_DER}) for the DER.1 interdependent system. We observe similar insights to SCADA's results (Section~\ref{sec:evaluation}) and the remaining systems.

\begin{figure}
\begin{tabular}{cc}
  \includegraphics[width=40mm]{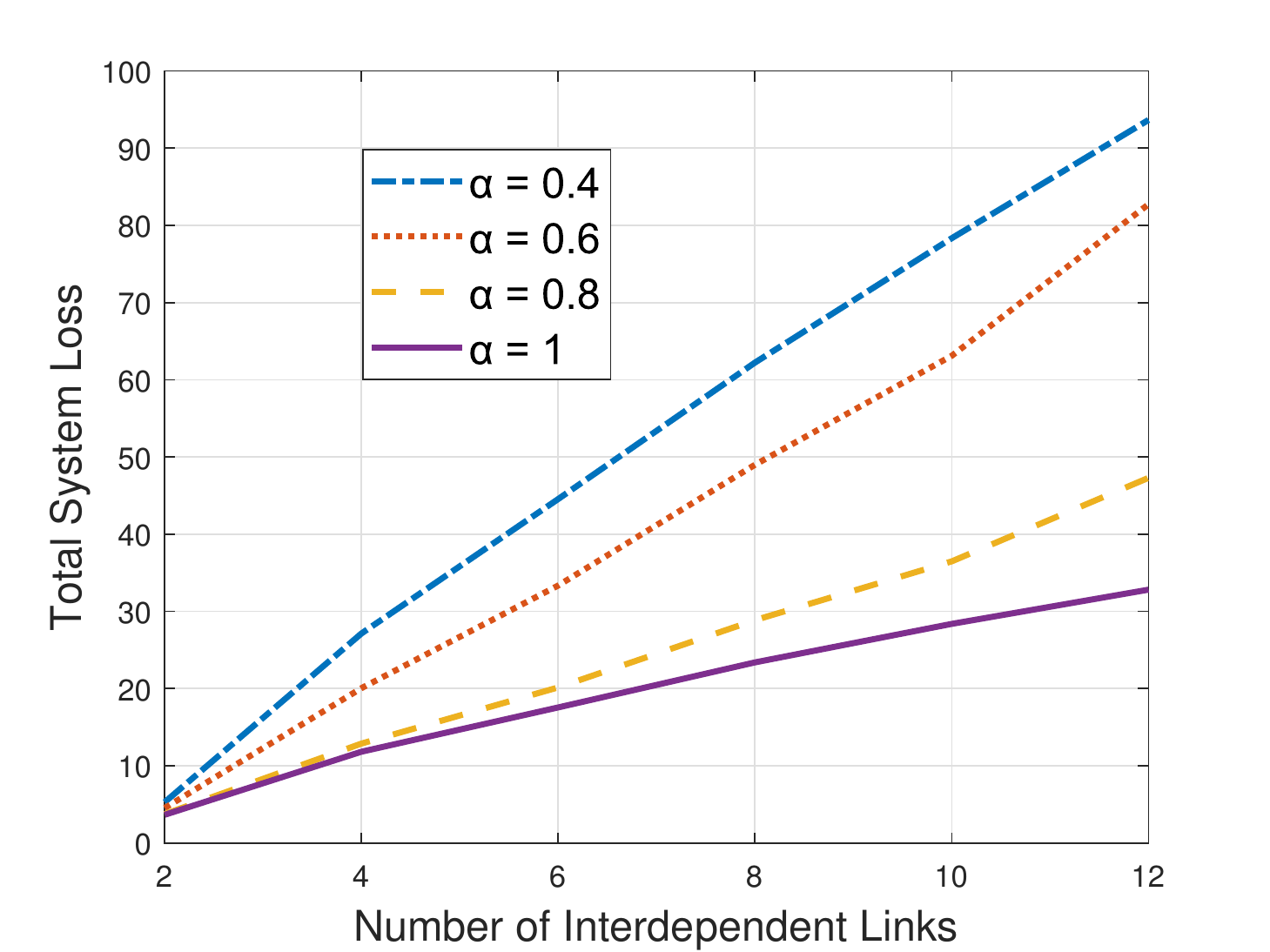} & \includegraphics[width=40mm]{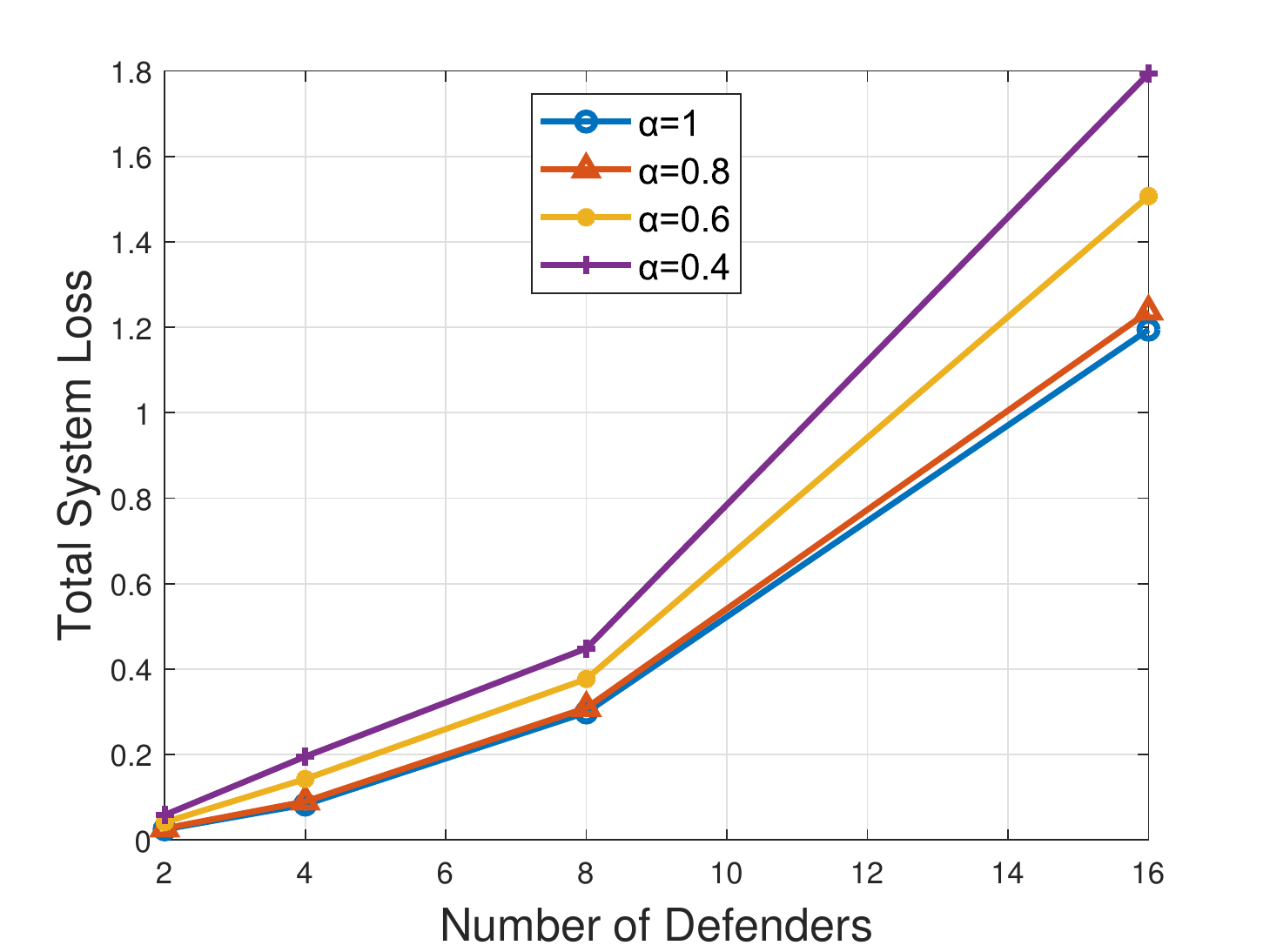} \\
(a) Interdependency Effect & (b) Number of Defenders \\[6pt]
 \includegraphics[width=40mm]{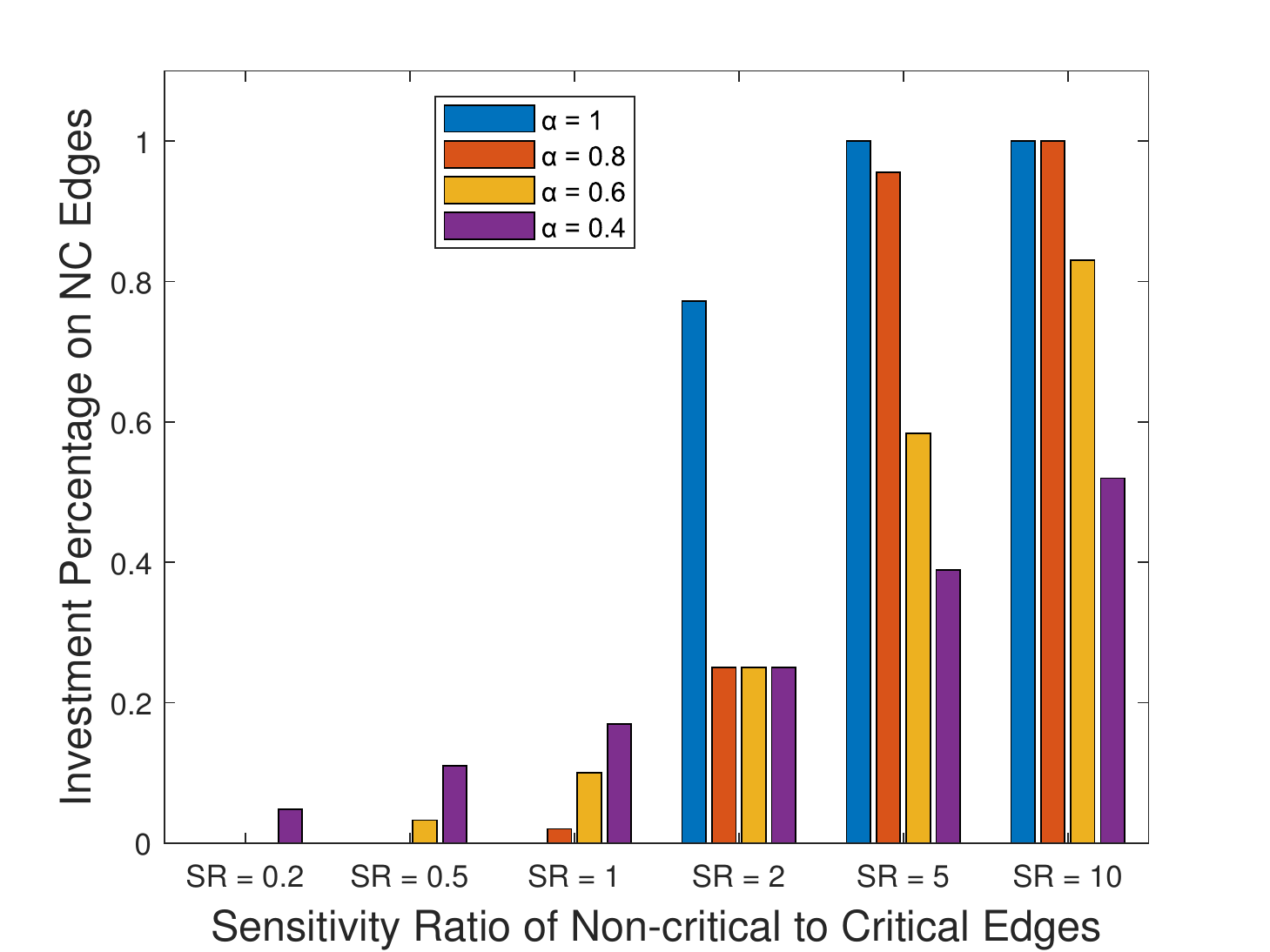} &   \includegraphics[width=40mm]{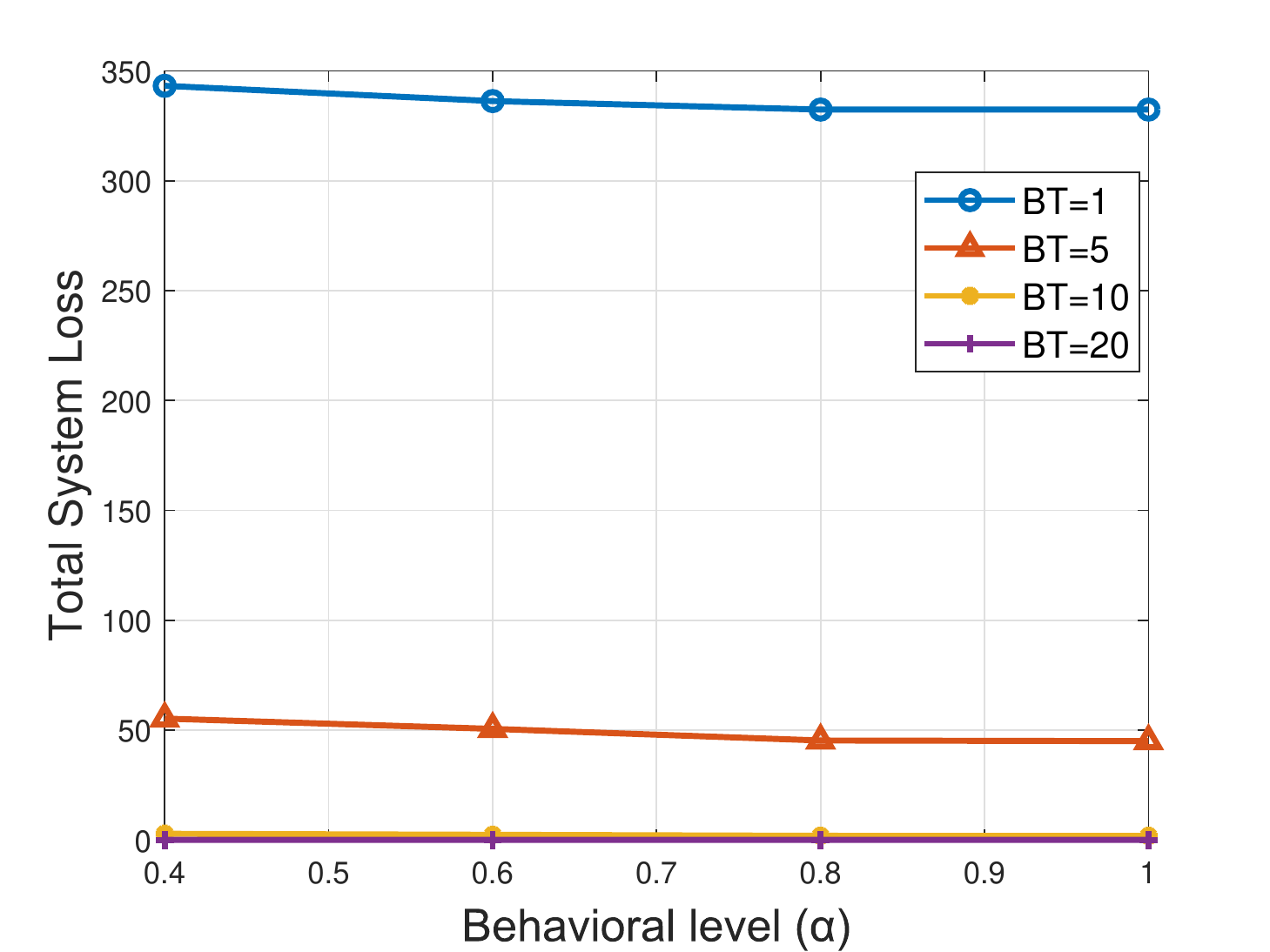} \\
(c) Sensitivity of Edges & (d) Security Budget \\[6pt]
\end{tabular}
\caption{Results of Multi-defenders for DER.1 system.}
\label{fig:extended_DER}
\end{figure}

\begin{figure}[t]
\begin{minipage}[t]{.45\textwidth}
\centering
   \includegraphics[width=0.9\linewidth]{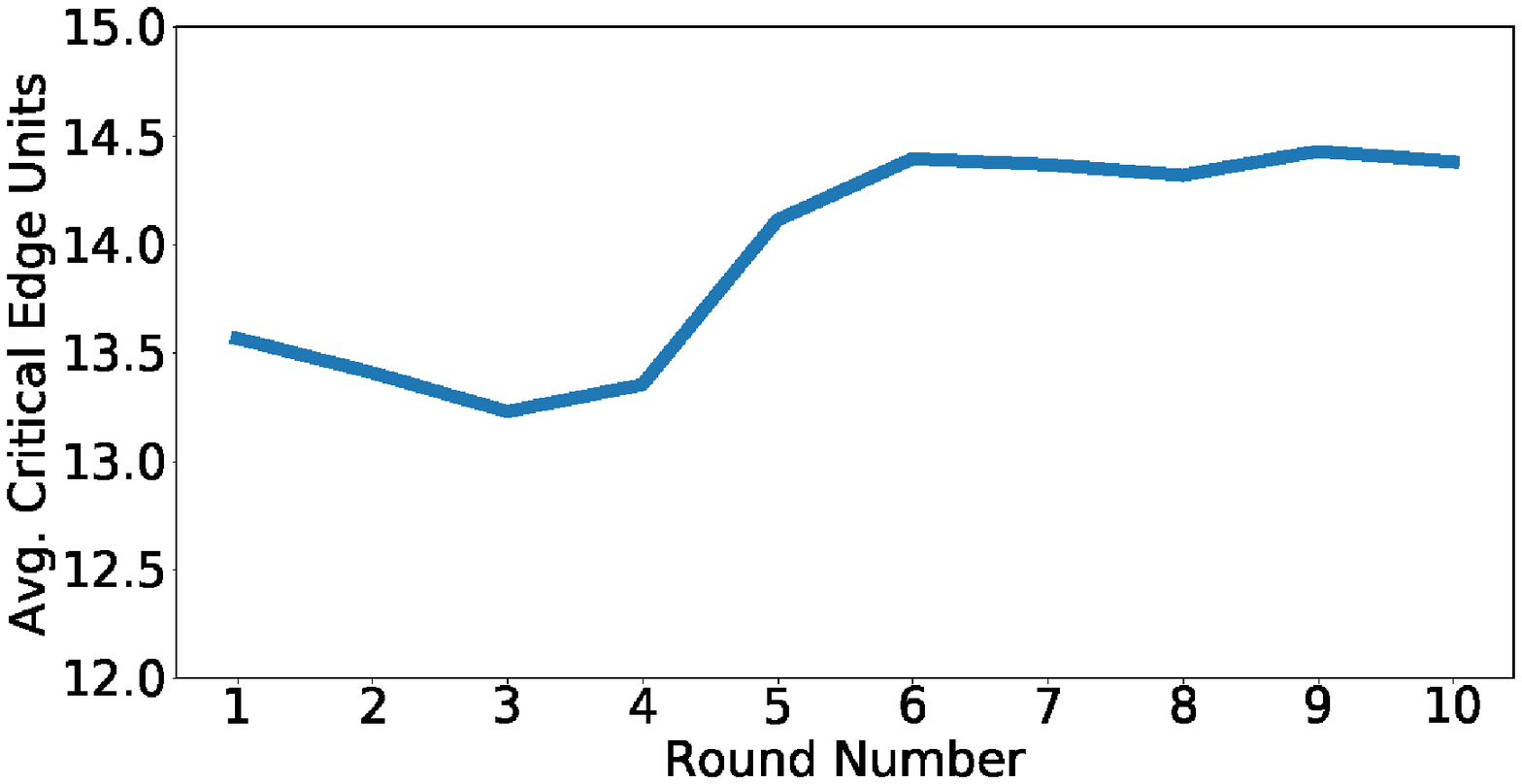}
  \caption{Average of all subjects' investments on the critical edge vs experiment rounds. The upward trend indicates that on average, subjects are learning.}
  \label{fig:agg_multirounds}
\end{minipage}
\begin{minipage}[t]{.45\textwidth}
\centering
   \includegraphics[width=0.9\linewidth]{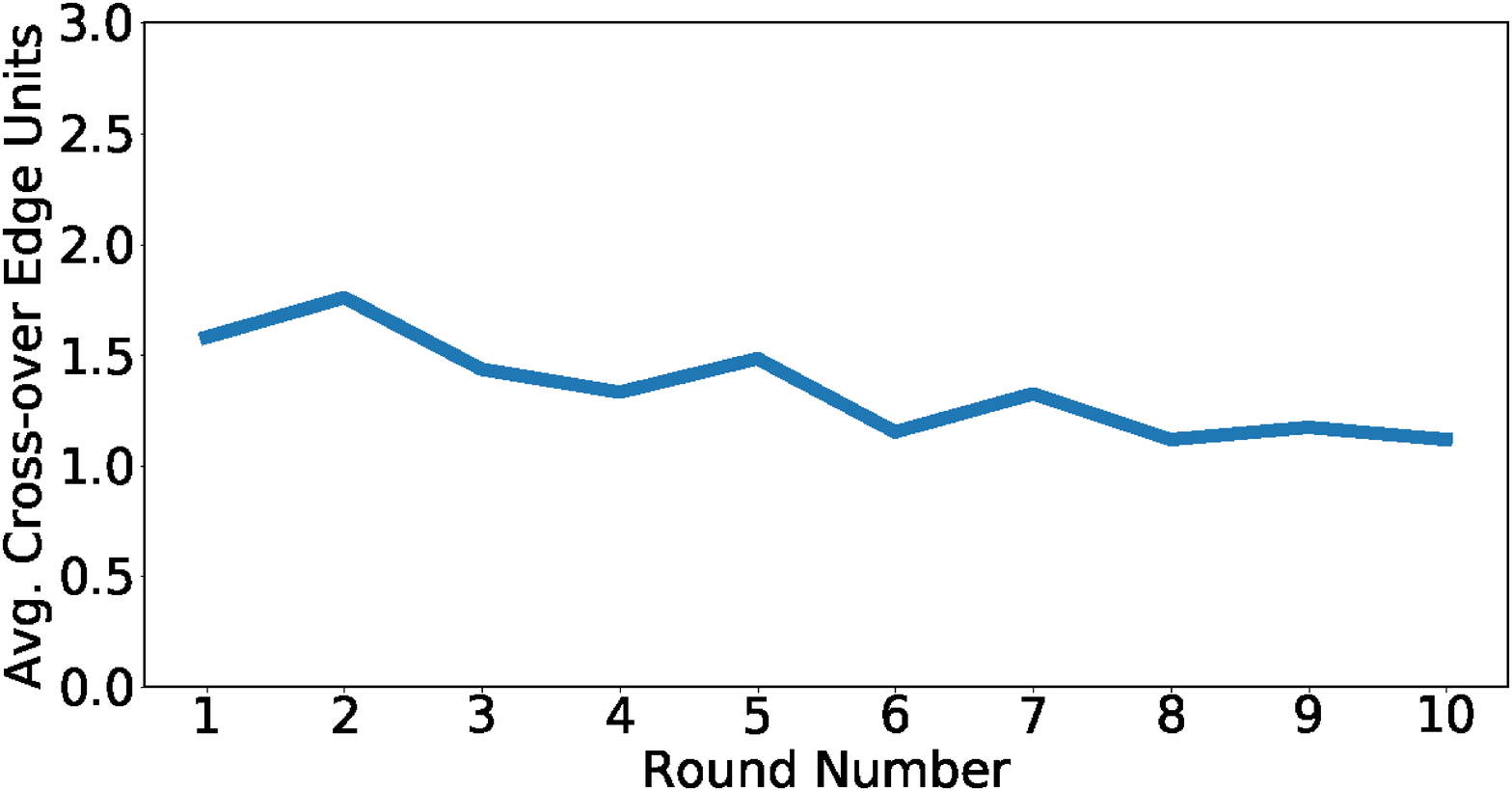}
  \caption{Average of all subjects' investments on the cross-over edge vs experiment rounds. There is only a weak  downward trend in spreading behavior.}
  \label{fig:agg_multirounds_blue}
\end{minipage}
\end{figure}

\subsection{Average Investments of Multi-rounds}
Here, we show the average investments for each round for both of the attack graphs tested in our human subject study. 
\subsubsection{Emulating Reinforcement Learning:}
Note that we emulated partially the reinforcement learning environment where in each round after the subject allocates her investments, a simulated attack is run and we show the subject if the critical asset was compromised or not and give her experimental points if she successfully defended the asset.

\end{document}